%% file: ms.tex
\documentclass[runningheads]{llncs}
\usepackage{graphicx}
\usepackage{sansmathaccent}
\usepackage{amsmath}
\usepackage{amssymb}
\usepackage{adjustbox}
\usepackage{stmaryrd}
\usepackage{mathtools}
\usepackage{enumerate}
\usepackage{comment}
\pdfmapfile{+sansmathaccent.map}

\usepackage[T1]{fontenc}
\IfFileExists{luximono.sty}{\usepackage[scaled=0.9]{luximono}}{\usepackage[scaled=0.81]{beramono}}
\usepackage{xcolor}
\usepackage{todonotes}
\usepackage{xspace}
\usepackage{semantic}
\input{macros}

\spnewtheorem{assumption}{Assumption}{\bfseries}{\rmfamily}

\begin{document}
\title{Flexible Refinement Proofs in Separation Logic}
%
%
\author{Aurel B\'il\'y \and
Christoph Matheja \and
Peter M{\"u}ller}
\authorrunning{A. B\'il\'y et al.}
%
\institute{Department of Computer Science, ETH Zurich, Switzerland\\
\email{\{aurel.bily, cmatheja, peter.mueller\}@inf.ethz.ch}}
\maketitle              
\begin{abstract}
Refinement transforms an abstract system model into a concrete, executable program, such that properties established for the abstract model carry over to the concrete implementation. Refinement has been used successfully in the development of substantial verified systems. Nevertheless, existing refinement techniques have limitations that impede their practical usefulness. Some techniques generate executable code automatically, which generally leads to implementations with sub-optimal performance. Others employ bottom-up program verification to reason about efficient implementations, but impose strict requirements on the structure of the code, the structure of the refinement proofs, as well as the employed verification logic and tools.

In this paper, we present a novel refinement technique that removes these limitations. Our technique uses separation logic to reason about efficient concurrent implementations. It prescribes only a loose coupling between an abstract model and the concrete implementation. It thereby supports a wide range of program structures, data representations, and proof structures.  We make only minimal assumptions about the underlying program logic, which allows our technique to be used in combination with a wide range of logics and to be automated using off-the-shelf separation logic verifiers. We formalize the technique, prove the central trace inclusion property, and demonstrate its usefulness on several case studies.
\end{abstract}

\section{Introduction}\label{sec:introduction}
\input{introduction}

\section{Overview}\label{sec:overview}
\input{overview}

\section{Preliminaries: Concurrent Separation Logic}\label{sec:preliminaries}
\input{preliminaries}

\section{Methodology and Formalization}\label{sec:methodology}

\input{methodology}

\section{Evaluation}\label{sec:evaluation}
\input{evaluation}

\section{Related Work}\label{sec:related_work}
\input{related_work}

\section{Conclusion}\label{sec:conclusion}
\input{conclusion}


%
%
\bibliographystyle{splncs04}
\bibliography{references}
\newpage
\appendix
\input{appendix}

\end{document}

%% file: macros.tex
\usepackage[normalem]{ulem} 

\makeatletter
\def\operator#1{\@ifnextchar\bgroup {\operatorarg{\ensuremath{#1}}}{\ensuremath{#1}}}
\def\operatorarg#1#2{{#1}{\ensuremath{(#2)}}}
\def\spoperator#1#2{\@ifnextchar\bgroup{\spoperatorarg{\ensuremath{#1}}{\ensuremath{#2}}}{\ensuremath{#1}}}
\def\spoperatorarg#1#2#3{\ensuremath{#1#2#3}}

\def\fixedoperator#1{\@ifnextchar\bgroup {\fixedoperatorarg{#1}}{\ensuremath{#1}}}
\def\fixedoperatorarg#1#2{\fixedoperatorparse{#1}#2~}
\def\fixedoperatorparse#1#2,#3~{\ensuremath{{#2}{.}{#1}{(#3)}}}
\makeatother

\newskip \point \point =1pt
\setbox134=\hbox{\leavevmode\raise0\point\hbox{${\langle}\kern-2.5\point{\langle}$}}
\setbox135=\hbox{\raise0\point\hbox{${ \rangle}\kern-2.5\point{ \rangle}$}}

\setbox136=\hbox{\leavevmode\raise0\point\hbox{${\lfloor}\kern-3.0\point{\lfloor}$}}
\setbox137=\hbox{\raise0\point\hbox{${ \rfloor}\kern-3.0\point{ \rfloor}$}}

\setbox138=\hbox{\leavevmode\raise0\point\hbox{${[}\kern-1.5\point{[}$}}
\setbox139=\hbox{\raise0\point\hbox{${]}\kern-1.5\point{]}$}}

\definecolor{darkred}{rgb}{0.55, 0.0, 0.0}

\newcommand\secref[1]{Sec.~\ref{#1}} 
\newcommand\appref[1]{App.~\ref{#1}}
\newcommand\defref[1]{Def.~\ref{#1}}
\newcommand\lemref[1]{Lemma~\ref{#1}}

\newcommand\tabref[1]{Tab.~\ref{#1}}
\newcommand\figref[1]{Fig.~\ref{#1}}
\newcommand\thmref[1]{Thm.~\ref{#1}}
\renewcommand\eqref[1]{(\ref{eq:#1})}
\newcommand\assref[1]{Assumption~\ref{#1}}

\newcommand\saferef[1]{\defref{def:conf-safe}.#1}
\newcommand\itref[1]{\textit{(\ref{item:#1})}}

\newcommand{\eg}{{{e.g.\@}}}

\newcommand{\ie}{{{i.e.\@}}}
\newcommand{\cf}{{{cf.\@}}} 

\newcommand{\TLA}{{{TLA$^{+}$}}}

\newcommand{\EndDef}{\hfill$\triangle$}
\newcommand{\Math}[1]{\ensuremath{#1}\xspace}
\newcommand{\defeq}{\Math{\triangleq}}
\newcommand{\ddefeq}{~\defeq~}
\newcommand{\Setfont}[1]{\Math{\textnormal{\textbf{#1}}}}
\newcommand{\Nats}{\Math{\mathbb{N}}}

\newcommand{\Set}[1]{\left\{\, #1 \,\right\}}
\newcommand{\fpto}{{}\rightharpoonup_{\text{fin}}{}}
\newcommand{\dom}[1]{\Math{\textit{dom}\left(#1\right)}}
\newcommand{\mydot}{~{}.{}~}
\newcommand{\undefined}{\textnormal{undefined}}
\newcommand{\tif}{~\textnormal{if}~}

\newcommand{\qiff}{\quad\textnormal{iff}\quad}

\newcommand{\qand}{\quad\textnormal{and}\quad}
\newcommand{\tand}{~\textnormal{and}~}

\newcommand{\qimplies}{\quad\textnormal{implies}\quad}

\newcommand{\eeq}{~{}={}~}

\newcommand{\cceq}{~{}::={}~}

\newcommand{\EE}{\Math{E}}

\newcommand{\Append}[2]{\Math{#2\, \texttt{++}\, #1}}

\newcommand{\Val}{\Math{v}}
\newcommand{\Vals}{\Setfont{Vals}}

\newcommand{\Var}{\Math{x}}
\newcommand{\Vars}{\Setfont{Vars}}

\newcommand{\BB}{\EE}
\newcommand{\True}{\texttt{true}}
\newcommand{\False}{\texttt{false}}

\newcommand{\Lock}{\Math{\mathcal{L}}}
\newcommand{\CLock}{\Math{\mathcal{R}}}
\newcommand{\Locks}{\Setfont{Locks}}
\newcommand{\GhostLock}{\Math{\mathcal{G}}}
\newcommand{\InitLock}{\Math{\mathcal{I}}}

\newcommand{\DLocks}[1]{\Math{\Locks(#1)}}

\newcommand{\Command}{\Math{C}}
\newcommand{\CC}{\Command}
\newcommand{\SimpleCommand}{\Command_{\text{base}}}
\newcommand{\Commands}{\Setfont{Cmds}}

\newcommand{\CommandFont}[1]{\Math{\textnormal{\texttt{#1}}}}
\newcommand{\Skip}{\CommandFont{skip}}
\newcommand{\Assign}[2]{#1~\CommandFont{:=}~#2}
\newcommand{\Deref}[1]{\Math{[#1]}}
\newcommand{\Mutation}[2]{\Assign{\Deref{#1}}{#2}}
\newcommand{\Lookup}[2]{\Assign{#1}{\Deref{#2}}}
\newcommand{\New}[2]{\Assign{#1}{\CommandFont{new}(#2)}}
\newcommand{\Free}[1]{\CommandFont{free}(#1)}
\newcommand{\Seq}{\Math{\mathrel{;}}}
\newcommand{\Par}{\parallel}
\newcommand{\If}[1]{\CommandFont{if}~#1~}
\newcommand{\Else}{~\CommandFont{else}~}
\newcommand{\Ite}[3]{\If{#1}\Block{#2}\Else\Block{#3}}
\newcommand{\While}[1]{\CommandFont{while}~#1~}
\newcommand{\Block}[1]{\{\, #1 \,\}}
\newcommand{\Print}[1]{{\CommandFont{print}}(#1)}
\newcommand{\Next}{{\CommandFont{Next}}}
\newcommand{\Initialize}{{\CommandFont{Init}}}
\newcommand{\NextBlock}[1]{\Next~\Block{#1}}
\newcommand{\InitializeBlock}[1]{\Initialize~\Block{#1}}
\newcommand{\DeclareLock}[1]{\CommandFont{lock}~#1~}
\newcommand{\DeclareLockBlock}[2]{\DeclareLock{#1}\Block{#2}}
\newcommand{\With}[2]{\CommandFont{with}~#1~\CommandFont{when}~#2~}
\newcommand{\WithBlock}[3]{\With{#1}{#2}\Block{#3}}
\newcommand{\Within}[1]{{\CommandFont{within}}~#1~}
\newcommand{\WithinBlock}[2]{\Within{#1}\Block{#2}}

\newcommand{\Locs}{\Setfont{Addrs}}
\newcommand{\Loc}{\Math{a}}
\newcommand{\ALocs}{\Setfont{GhostAddrs}}

\newcommand{\Stack}{\Math{s}}
\newcommand{\Heap}{\Math{h}}

\newcommand{\HPerm}[1]{\Math{\textsf{perm}(#1)}}
\newcommand{\HVal}[1]{\Math{\textsf{val}(#1)}}

\newcommand{\Stacks}{\Setfont{Stacks}}
\newcommand{\Heaps}{\Setfont{NormHeaps}}
\newcommand{\EmptyHeap}{\Math{\Heap_{\emptyset}}}
\newcommand{\Add}{\Math{\oplus}}

\newcommand{\Predicate}[2]{\Math{\textit{#1}\left(#2\right)}}
\newcommand{\isDefined}[1]{\Predicate{def}{#1}}

\newcommand{\Abort}{\CommandFont{abort}}

\newcommand{\Configuration}{\Math{c}}
\newcommand{\Confs}{\Setfont{Confs}}
\newcommand{\Conf}[3]{#1, #2, #3}

\newcommand{\Rule}[1]{(\textnormal{\textsc{#1}})\xspace}

\newcommand{\Sem}[1]{#1}

\newcommand{\Update}[2]{\Math{\left[#1\,\textnormal{\texttt{:=}}\,#2\right]}}
\newcommand{\Delete}[1]{\Update{#1}{\bot}}

\newcommand{\stdout}{\Math{\textnormal{\texttt{stdOut}}}}
\newcommand{\OutVar}{\stdout}

\newcommand{\Locked}[1]{\Setfont{locked}(#1)}
\newcommand{\Reads}[2]{\Setfont{accesses}(#1, #2)}
\newcommand{\Writes}[2]{\Setfont{writes}(#1, #2)}

\newcommand{\OpRel}{{}\Math{\to}{}}
\newcommand{\NOpRel}{{}\Math{\not\to}{}}

\newcommand{\OpRelTrans}{{}\Math{\to^{*}}{}}

\newcommand{\isAtomic}[1]{\Predicate{atomic}{#1}}
\newcommand{\isInit}[1]{\textit{init}\left(#1\right)}

\newcommand{\InitialConfs}{\Setfont{I}}
\newcommand{\Traces}[1]{\Math{\Setfont{Traces(}#1\Setfont{)}}}

\newcommand{\Emptyseq}{\Math{\left[ ~ \right]}}

\newcommand{\Obs}[1]{\textit{obs}(#1)}

\newcommand{\PermsSym}{\Setfont{Perms}}
\newcommand{\Perms}{[0,1]}
\newcommand{\PermAdd}{\oplus}
\newcommand{\PermWrite}{1}

\newcommand{\PermNone}{0}
\newcommand{\Perm}{\rho}

\newcommand{\PHeaps}{\Setfont{Heaps}}
\newcommand{\PHeap}{\Math{\mathfrak{h}}}

\newcommand{\SL}{\Setfont{SL}}
\newcommand{\FOL}{\Setfont{FOL}}
\newcommand{\Assertion}[1]{\Math{#1}}
\newcommand{\PP}{\Assertion{P}}
\newcommand{\PQ}{\Assertion{Q}}
\newcommand{\PR}{\Assertion{R}}
\newcommand{\PI}{\Assertion{I}}

\newcommand{\emp}{{}\textnormal{\texttt{emp}}{}}
\newcommand{\acc}[2]{{}#1 \xhookrightarrow{#2} \_{}}
\newcommand{\alloc}[2]{{}#1 \xmapsto{#2} \_{}}
\newcommand{\pt}[3]{{}#1 \xmapsto{#2} #3{}}
\newcommand{\apt}[3]{{}#1 \xhookrightarrow{#2} #3{}}

\newcommand{\Implies}{{}\Rightarrow{}}
\newcommand{\sep}{{}\mathrel{*}{}}
\newcommand{\mw}{-\kern-.6em\raisebox{-.659ex}{*}\ }
\newcommand{\ex}[1]{\exists #1 \mydot}
\newcommand{\fa}[1]{\forall #1 \mydot}

\newcommand{\valid}{\Math{\models}}

\newcommand{\Absfont}[1]{\Math{\textnormal{\textit{#1}}}}
\newcommand{\AbsTS}{\Absfont{ATS}}
\newcommand{\AbsLength}{\Bbbk}
\newcommand{\AbsVars}{\vec{x}}
\newcommand{\AbsVar}[1]{x_{#1}}
\newcommand{\AbsInit}{\Absfont{Init}}
\newcommand{\AbsNext}{\Absfont{Next}}
\newcommand{\Primed}[1]{#1'}
\newcommand{\Abs}[1]{\left\lfloor #1\right\rfloor}
\newcommand{\AbsLoc}[1]{\Abs{x_{#1}}}
\newcommand{\Aabsvars}[1]{\Abs{\AbsVars}\!(#1)}
\newcommand{\AbsLocsOfName}{\textit{get\_state}}
\newcommand{\AbsLocsOf}[1]{\AbsLocsOfName(#1)}
\newcommand{\AbsLocs}{\Abs{x_1},\ldots,\Abs{x_\AbsLength}}

\newcommand{\AbsStates}{\Math{\Sigma_{\AbsTS}}}
\newcommand{\AbsState}{\sigma}

\newcommand{\AbsPaths}{\Setfont{Paths}(\AbsTS)}
\newcommand{\AbsTraces}{\Traces{\AbsTS}}

\newcommand{\FV}[1]{\Setfont{FV}(#1)}
\newcommand{\Mod}[1]{\Setfont{Mod}(#1)}
\newcommand{\subst}[2]{[#1/#2]}

\newcommand{\LockDef}[2]{\Math{#1\colon #2}}
\newcommand{\LockEnv}{\Math{\Gamma}}
\newcommand{\LockEnvG}{\Math{\Gamma,\LockDef{\GhostLock}{\GhostInv}}}

\newcommand{\Ghost}[1]{\Math{{#1}}}
\newcommand{\GhostInv}{\Ghost{G}}

\newcommand{\refineIndicator}{\mathcal{R}}
\newcommand{\triple}[5]{#2\:\vdash\:\left\{ #3 \right\}\:#4\:\left\{ #5 \right\}}
\newcommand{\rtriple}[5]{#2\:\vdash_{\refineIndicator}\:\left\{ #3 \right\}\:#4\:\left\{ #5 \right\}}

\newcommand{\striple}[3]{\triple{}{\LockEnv}{#1}{#2}{#3}}
\newcommand{\srtriple}[3]{\rtriple{}{\LockEnv}{#1}{#2}{#3}}

\newcommand{\semtriple}[5]{#2\models_{#1}\left\{ #3 \right\}#4\left\{ #5 \right\}}
\newcommand{\rsemtriple}[4]{#1\models_{\refineIndicator}\left\{ #2 \right\}#3\left\{ #4 \right\}}

\newcommand{\Safe}[7]{\Math{\textsf{safe}_{#2}\llbracket\,#1\,\rrbracket\left(#3,#4,#5,#6,#7\right)}}

\newcommand{\Succ}{\Math{\textsf{succ}}}

\newcommand{\RSucc}{\Math{\textsf{refsucc}}}
\newcommand{\HeapF}{\Math{\PHeap_{\textnormal{F}}}}
\newcommand{\HeapS}{\Math{\PHeap_{\textnormal{S}}}}

\newcommand{\SoundInitialsFor}[1]{\Setfont{I}(#1,\PP,\LockEnv)}
\newcommand{\SoundInitials}{\Setfont{I}(\CC,\PP,\LockEnv)}

%% file: introduction.tex
Refinement is a powerful technique for the formal development of correct systems. It is especially useful for concurrent and distributed systems, because it allows one to establish system-wide invariants on the level of abstract models and preserve them when decomposing the system into its components. It is, thus, not surprising that several recent developments of verified systems employ refinement reasoning~\cite{DeepSpecPosition,HawblitzelHKLPR15,KleinEHACDEEKNSTW09,Leroy06,LorchCKPQSWZ20}.

Traditionally, refinement is applied to mathematical models of software, for instance, in formalisms such as Event-B~\cite{Abrial10}, \TLA~\cite{Lamport2002}, or the higher-order logics supported by interactive theorem provers such as Coq~\cite{coq}. The final executable program is then produced automatically by a code generator. However, this approach generally leads to sub-optimal implementations that do not fully utilize the language features needed to produce efficient code, such as mutable heap structures and concurrency. Manually optimizing the generated code would forfeit the correctness guarantees provided by the formal development.

To address these shortcomings, recent work has combined refinement with bottom-up program verification techniques that support more features of modern programming languages~\cite{HawblitzelHKLPR15,Koh0LXBHMPZ19,Igloo,Trillium}. While these approaches substantially increase the range of programs that can be developed using refinement, they come with their own limitations. First, several existing techniques restrict the structure of the executable program, which reduces expressiveness and limits the efficiency of the executable code. For instance, the methodology used in IronFleet~\cite{HawblitzelHKLPR15} supports protocol-level concurrency, but restricts the implementations of individual components to execute sequentially. Moreover, the structure of the code must closely follow the structure of the abstract \TLA model. Similarly, the Igloo methodology~\cite{Igloo} does not allow threads to perform I/O operations concurrently. Second, most refinement techniques that support bottom-up code development are closely tied to a particular program logic, which often impedes adoption 
 and automation. For example, refinement in DeepSpec~\cite{Koh0LXBHMPZ19} is tied to the VST logic~\cite{CaoBGDA18}, DISEL~\cite{SergeyWT18} comes with its own dedicated logic, and Trillium~\cite{Trillium} leverages Iris~\cite{Jung2018}. All of these logics are very expressive, but require substantial manual effort. The underlying program logic may also impose limitations. For instance, Trillium inherits Iris's restriction to finitary behaviors, which precludes operations such as non-deterministically choosing from an infinite set. However, non-deterministic choice is essential for the specification of abstract system models, for instance, when the concrete algorithm to determine a value is an implementation decision that is taken in a later refinement step.

In this paper, we present a novel methodology to prove that an implementation refines an abstract model given as a transition system. Our methodology enables \emph{flexible} refinement proofs along four dimensions:

\begin{enumerate}
\item \emph{Abstract model.} We do not prescribe a specific formalism for abstract models, but support any transition system whose transition relation can be specified in first-order logic.

\item \emph{Program structure.} Our methodology uses separation logic~\cite{Reynolds2002} to support efficient implementations, in particular, mutable state and arbitrary concurrency structures.

\item \emph{Logic and automation.} We make only minimal assumptions about the underlying program logic, which allows our methodology to be used in combination with a wide range of logics and to be automated using off-the-shelf separation logic verifiers. 

\item \emph{Proof structure.} Our methodology prescribes only a very loose coupling between the abstract model and the concrete implementation. This maximizes flexibility when choosing the program structure (for instance, control flow, concurrency, and thread synchronization), the data representation (supporting for instance, local and shared state), and proof structure (for instance, allowing coupling invariants to be expressed via a combination of local assertions and lock invariants).
\end{enumerate}

\paragraph{Approach.}

Our goal is to prove that a given implementation refines an abstract model, that is, each finite trace of the abstract model corresponds to a trace in the transition system. This trace inclusion property guarantees that any safety property proved for the abstract model also holds for the implementation.

Our abstract models are expressed as (possibly infinite) transition systems in the style of \TLA, Event-B, or other refinement frameworks. We assume that some variables of the abstract transition system (\AbsTS) are declared to be part of the environment in which the system executes. These are the variables that can be observed about an execution; that is, we require trace inclusion only after projecting traces to the environment variables. Consequently, a concrete implementation may perform arbitrary internal operations as long as it manipulates the environment according to \AbsTS. A typical example for environment variables are the input and output streams of a program, but our methodology also supports other cases such as the variables shared between threads.  We assume that a program manipulates the environment only via dedicated operations, such as the methods of an I/O library, whose specifications express how they manipulate the environment in terms of the environment variables of \AbsTS.

To reason about \AbsTS and the concrete implementation within one program logic, we embed \AbsTS as global ghost state into the concrete implementation: each variable of \AbsTS is represented as a ghost variable of the program. 
The concrete program must implement the transitions of \AbsTS via atomic program operations and corresponding updates of the ghost state. We identify these operations via annotations in the code; they include in particular the operations used to manipulate the environment. Proof obligations enforce that each such operation performs a valid transition of \AbsTS and that they execute atomically. We do not prescribe how these atomicity checks are performed: they are trivial for sequential programs, could be performed syntactically by identifying a set of atomic operations (such as compare-and-swap and native I/O operations), could be ensured via a global locking strategy, or could be discharged in a logic that can reason about atomicity, such as TaDA~\cite{PintoDG14}. This approach allows us to use any separation logic that is able to reason about shared state, in particular, standard separation logic~\cite{Reynolds2002}, concurrent separation logic~\cite{OHearn04}, and the numerous separation logics for fine-grained concurrency, e.g.,~\cite{Dinsdale-YoungDGPV10,PintoDG14,SergeyNB15}.

A key virtue of our methodology is that refinement proofs are organized around a minimal core, namely the proof obligations showing that the program manipulates the environment variables according to the abstract model. Discharging these proof obligations generally requires suitable coupling invariants between the abstract and the concrete state, as well as updates to the ghost state to maintain them. However, we do not prescribe how to express and prove those, which enables programmers to flexibly structure the code, data, and proof.

\paragraph{Contributions.}

We make the following technical contributions:

\begin{itemize}
\item We present a novel verification methodology for refinement proofs that provides more flexibility than prior work in terms of the supported class of programs, the choice of verification logic and tools, and the organization of the refinement proof itself (\secref{sec:overview}). 

\item We formalize an instance of our methodology for a concurrent language and a standard separation logic. Our soundness proof shows trace inclusion, which implies that safety properties of the abstract model also hold for the concrete implementation (\secref{sec:preliminaries} and \secref{sec:methodology}).

\item We demonstrate the expressiveness of our methodology by encoding a series of interesting examples into the Viper language~\cite{MuellerSchwerhoffSummers16}. Our evaluation shows in particular that our methodology can be automated with existing separation logic verifiers (\secref{sec:evaluation}). We will submit our examples as an artifact.
\end{itemize}

%% file: overview.tex
In this section, we explain our methodology on a simple example that prints all integers in ascending order, starting from an arbitrary, non-negative initial value. To illustrate the flexibility of our methodology, we refine an abstract model into a concurrent implementation. The concrete state is stored in local variables of the individual threads, which requires a non-trivial, decentralized coupling invariant. This section provides an informal overview; the details of the programming language and proof rules will be formalized in the next sections.

\begin{figure}[b]
\[
\begin{array}{l}
\texttt{Vars:} \;	\texttt{count: Int,~stdOut: Seq[Int]}\\
\texttt{Init:} \;0 \leq 	\texttt{count} \wedge 	\texttt{stdOut} = []\\
\texttt{Next:} \;	\texttt{stdOut}' = 	\texttt{stdOut ++ [count]} \wedge
\texttt{count}' = 	\texttt{count} + 1
\end{array}
\]
\caption{Abstract model of our example. Here, unprimed and primed variables refer the variables values before and after the transition, respectively.}
\label{fig:example-abstract}
\end{figure}

\paragraph{Abstract Model.}
The abstract model of our system is defined in \figref{fig:example-abstract}. We represent the output stream of the system \texttt{stdOut} as a sequence of integers. This variable belongs to the environment, whereas \texttt{count} is internal to the system. The initial state has an arbitrary non-negative value for \texttt{count}, which illustrates the common case that abstract models choose values non-deterministically. A valid transition 
adds the current value of \texttt{count} to the output stream and then increments it.

\paragraph{Embedding of the Abstract Model.}
As explained in the introduction, we embed the abstract model as ghost state into the concrete implementation. To this end, we declare a global ghost variable:

\texttt{ghost var count: Int}

\noindent
We assume that the environment variable \texttt{stdOut} is predeclared for each program. It is manipulated via a \texttt{print(x)} statement, which requires ownership of \texttt{stdOut} (in the sense of separation logic: $\alloc{\texttt{stdOut}}{}$) and appends its argument \texttt{x} to the output stream. Other input and output channels, as well as other I/O operations are handled analogously.

To refine the abstract model, the concrete implementation may manipulate \texttt{count} and \texttt{stdOut} only according to the abstract model presented in \figref{fig:example-abstract}. We enforce this requirement by protecting both variables with a \emph{ghost lock}. Like a regular lock, a ghost lock is equipped with a lock invariant that must be established when the lock is initialized and must be preserved by the operations between an acquire and a release operation. In contrast to a regular lock, a ghost lock is not present at run time; that is, it does not block execution and, thus, cannot cause deadlock. Erasing ghost locks from the program is sound because the operations between an acquire and a release must execute atomically.

In our language, we initialize the ghost lock with a dedicated \texttt{Init} ghost statement, which checks the \texttt{Init} assertion from the abstract model as well as the lock invariant (in particular, it transfers ownership of the variables in the lock invariant from the executing thread to the ghost lock). A dedicated \texttt{Next} ghost block statement acquires the ghost lock, executes the block of the statement, and then releases the ghost lock. Upon release, it checks that the \texttt{Next} assertion from the abstract model holds between the state in which the ghost lock was acquired and the state where it is released. That is, the \texttt{Next} statement executes \emph{one} transition of the abstract model (or stutters). The block of a \texttt{Next} statement must be atomic; we enforce this requirement syntactically, that is, we check that the body consists of at most one atomic statement plus an arbitrary number of ghost operations. However, our methodology also supports more sophisticated approaches, for instance, logics that can prove atomicity~\cite{PintoDG14}. Note that both \texttt{Init} and \texttt{Next} are ghost statements, that is, part of the specification. The executable program contains their blocks, but not the manipulation of the ghost lock. They are supported by any separation logic that handles locks.

The invariant of the ghost lock must contain at least fractional ownership~\cite{Boyland03} to each variable of the abstract model, which ensures that once the ghost lock has been initialized, those variables can be modified \emph{only} within a \texttt{Next} statement and, thus, the \texttt{Next} relation of the abstract model is checked on each modification. In our example, the ghost lock invariant contains fractional ownership of the \texttt{count} variable and full ownership of \texttt{stdOut}:
\[
\alloc{\texttt{count}}{\frac{2}{3}} ~\sep~
\alloc{\texttt{stdOut}}{}
\]
Our proof rules enforce that \texttt{print}---the only way to modify $\texttt{stdOut}$---can be executed only after the ghost lock has been initialized with an \texttt{Init} statement. Consequently, the above lock invariant guarantees that \texttt{print} operations can occur only within a \texttt{Next} block and, thus, all changes to the environment are checked to comply with the abstract model. We assume that \texttt{print} is atomic, which is standard for I/O operations.

\begin{figure}[t]
\begin{align*}
 & \Assign{\texttt{count}}{0}\Seq \\[-1mm]
 & \Initialize~\{ \\[-1mm]
 & \qquad \Assign{\texttt{evenTurn}}{\True}\Seq
          \Assign{\texttt{lastEven}}{-1}\Seq
          \Assign{\texttt{lastOdd}}{0}\Seq \\[-1mm]
 & \qquad \DeclareLock{\Lock} \{ ~
        \text{/* even-thread */} \,\Par\, \text{/* odd-thread */ }~ \} \\[-1mm]
 & \}
\end{align*}
\vspace{-6mm}
\caption{The concrete implementation of our example. The even-thread is presented in \figref{fig:example-even}; the odd-thread is analogous.}
\label{fig:example-main}
\end{figure}

\figref{fig:example-main} shows the concrete implementation. After initializing the abstract state, the \texttt{Init} operation checks that the initial state is valid and creates the ghost lock that guards further updates to the abstract state, in particular, all executions of \texttt{print}. Note that our implementation fixes a concrete initial value for \texttt{count}, whereas the abstract model permits an arbitrary non-negative number (see \figref{fig:example-abstract}). Resolving non-determinism is very common during refinement. We explain the body of the \texttt{Init} statement below.

\paragraph{Concrete Program State and Coupling.}

Our concrete implementation uses two concurrent threads (parallel branches) that print the even and odd numbers, respectively. To synchronize these two threads, we declare a global Boolean variable that indicates whether the next number to be printed is even:

\texttt{var evenTurn: Bool}

\noindent
This variable is protected by a global (regular, non-ghost) lock $\Lock$. 

\begin{figure}[t]
\begin{align*}
  & \texttt{var}~\Assign{\texttt{c}}{0}\Seq \\[-3mm]
  & \While{\True}~\texttt{invariant~lastEven} \xmapsto{\frac{1}{2}} 2 * \texttt{c} - 1\\[-1mm]
  & \{ \\[-1mm]
  & \qquad\With{\Lock}{\texttt{evenTurn}}~\{ \\[-1mm]
  & \qquad\qquad\Next~\{\quad\texttt{/*}~\text{print + ghost update}~\texttt{*/} \\[-1mm]
  & \qquad \qquad\qquad \Print{2 * \texttt{c}}\Seq 
    \Assign{\texttt{count}}{\texttt{count}+1}~\}\quad \\[-1mm]
  & \qquad\qquad \Assign{\texttt{lastEven}}{\texttt{lastEven} + 2}\Seq \\[-1mm]
  & \qquad\qquad \Assign{\texttt{evenTurn}}{\texttt{!evenTurn}}\Seq \\[-1mm]
  & \qquad\qquad \Assign{\texttt{c}}{\texttt{c}+1}~\\[-1mm]
  & \qquad\}\quad\texttt{/*}~\text{release lock \Lock}~\texttt{*/} \\[-1mm]
  &  \}
\end{align*}
\vspace{-6mm}
\caption{Implementation of the thread that prints the even numbers. The \texttt{Next} block marks an atomic transition; its body is atomic because the \texttt{print} statement is atomic and the subsequent assignment is a ghost statement.}
\label{fig:example-even}
\end{figure}

\figref{fig:example-even} shows the implementation of the thread that prints the even numbers; the other thread is analogous. Both threads loop indefinitely. In each iteration, they acquire the lock $\Lock$ (waiting until \texttt{evenTurn} is true resp.\ false), execute a transition (explained later), flip \texttt{evenTurn}, and release the lock.
Each thread has a local variable \texttt{c} that counts how many numbers it has printed. This design illustrates that our methodology supports flexible combinations of global concrete state (such as \texttt{evenTurn}) and local concrete state (such as \texttt{c}). While global concrete state can easily be connected to the abstract state via lock invariants, local variables require more flexible ways of expressing the coupling invariant. 

For instance, at the beginning of each loop iteration of the even-thread, the abstract counter \texttt{count} is equal to $2*\texttt{c}$ in case \texttt{evenTurn} is true. This condition allows us to prove that the \texttt{print} operation in the loop body is indeed permitted by the abstract model. However, this coupling invariant cannot be included in a lock invariant because locks do not protect local variables. Nor can it be expressed as a loop invariant because that would require that the even-thread holds on to some ownership of \texttt{count}, which would prevent the odd-thread from ever updating it.

A relation between the local variable \texttt{c} and the shared variable \texttt{count} can be proved in many ways, for instance, by using classical rely-guarantee reasoning~\cite{Jones81} or concurrent abstract predicates~\cite{Dinsdale-YoungDGPV10}. Our refinement methodology is compatible with any such logic. In our example, we use a standard encoding of the Owicki-Gries counter~\cite{OwickiG76} in separation logic. For this purpose, we introduce two global ghost variables \texttt{lastEven} and \texttt{lastOdd} that keep track of the effect of each individual thread:

\texttt{ghost var lastEven: Int, lastOdd: Int}

\noindent
We relate these ghost variables to the local variable \texttt{c} in each thread via the thread's loop invariant (see \figref{fig:example-even}), and also to the global \texttt{count} via the lock invariant of lock $\Lock$:
\[
\begin{array}{l}
\pt{\texttt{evenTurn}}{}{\_} \sep
\pt{\texttt{lastEven}}{\frac{1}{2}}{\_} \sep
\pt{\texttt{lastOdd}}{\frac{1}{2}}{\_} \sep
\pt{\texttt{count}}{\frac{1}{3}}{\_} \sep\\
(\texttt{evenTurn} \Rightarrow 	\texttt{count} = 	\texttt{lastOdd} \wedge 	\texttt{lastEven} = 	\texttt{lastOdd} - 1) \sep \\
(\neg \texttt{evenTurn} \Rightarrow 	\texttt{count} = 	\texttt{lastEven} \wedge 	\texttt{lastOdd} = 	\texttt{lastEven} - 1)
\end{array}
\]
\noindent
The lock and the loop invariants \emph{together} form the coupling invariant for our example. For the even-thread, we get $\texttt{lastEven} = 2*c-1$ from the loop invariant, \texttt{evenTurn} from the \texttt{with}-statement, and 
$(\texttt{evenTurn} \Rightarrow 	\texttt{count} = 	\texttt{lastOdd} \wedge 	\texttt{lastEven} = 	\texttt{lastOdd} - 1)$ from the lock invariant. These three conditions together imply $	\texttt{count}=2*c$, which is required to show that the printed value is permitted by the abstract model.

\paragraph{Discussion.}
Our example illustrates that our methodology enables flexible refinement proofs, which are required to support a wide range of efficient implementations. We refined an abstract model into a concurrent implementation that uses both local and mutable shared state, as well as thread synchronization via locks. The proof makes only minimal assumptions about the underlying program logic. Concretely, we use concurrent separation logic, locks, ghost variables, and fractional permissions. These features are supported and automated by many existing separation logic verifiers. For instance, an encoding of our example into Viper verifies automatically in around 3.8s. Combining our refinement methodology with other, more advanced program logics is possible. 

Finally, our example demonstrates that our methodology enables flexible proof structures. Proofs are essentially derived backwards from those statements that manipulate environment variables, here, \texttt{print}. These statements require that the abstract state has been initialized and that the modification of the environment variables is permitted by the abstract model. Any proof structure that establishes these properties is compatible with our methodology. In our example, we use a combination of loop invariants and lock invariants, connected via global ghost variables, to establish the necessary coupling relation.
This flexibility is essential to support a wide range of data, control, and concurrency structures.

%% file: preliminaries.tex
Our verification technique for refinement proofs does not depend on a particular 
program logic but can, in principle, be integrated into most
separation logic-based verification techniques.
To make this claim more concrete, we will formalize (in \secref{sec:methodology})
our methodology on top of an elementary formalization of concurrent separation logic with fractional permissions
by Vafeiadis~\cite{Vafeiadis11}.
Since Vafeiadis' soundness proof generalizes well to more
advanced concurrent separation logics, such as \cite{Dinsdale-YoungDGPV10}, we expect
the same when using our technique for refinement on top
of such advanced logics.

This section briefly recapitulates the main ingredients of concurrent separation logic (CSL).
More precisely, we introduce a small concurrent programming language,
its underlying model of program states, and its operational semantics.
Furthermore, we discuss CSL's assertion language and proof rules.

\subsection{Programming Language}
\label{sec:sound:language}
We consider a small programming language that supports heap-manipulating instructions,
structured concurrency, and locks.
More precisely, the set \Commands of \emph{commands} in our programming language is
given by the grammar

\noindent
\begin{minipage}{0.5\textwidth}
\begin{align*}
    \Command \quad\cceq\quad 
       & \SimpleCommand \\
    |~ & \Command \Seq \Command \\
    |~ & \If{\BB} \Block{ \Command } \Else \Block{ \Command } \\
    |~ & \While{\EE} \Block{\Command} \\
    |~ & \Command \Par \Command \\
    |~ & \DeclareLock{\Lock}\Block{\Command} \\
    |~ & \With{\Lock}{\EE} \Block{ \Command } \\
    |~ & \Within{\Lock} \Block{ \Command }  semantics
\end{align*}
\end{minipage}\hfill%
\begin{minipage}{0.45\textwidth}
\begin{align*}
    \SimpleCommand \quad\cceq\quad 
       & \Skip \\
    |~ & \Assign{\Var}{\EE} \\
    |~ & \Mutation{\EE}{\EE} \\
    |~ & \Lookup{\Var}{\EE} \\
    |~ & \Free{\EE} \\
    |~ & \New{\Var}{\EE} \\
       & \\
       & 
\end{align*}
\end{minipage}

\medskip
\noindent
where $x$ is a \emph{variable} in the set $\Vars$,
$\EE$ is an \emph{expression} over variables, 
and $\Lock$ is a \emph{lock identifier} taken from the arbitrary, but finite 
 set $\Locks = \Set{\Lock,\ldots}$.

We briefly go over our language: 
in addition to the assignment $\Assign{x}{\EE}$ and the usual control-flow structures, 
$\Command_1\Par\Command_2$ is the \emph{parallel composition} of $\Command_1$ and $\Command_2$;
$\DeclareLockBlock{\Lock}{\Command}$ \emph{declares a new lock} $\Lock$ that can be used in $\Command$ and expires after termination of $\Command$;
the \emph{conditional critical region} (CCR)
$\WithBlock{\Lock}{\EE}{\Command}$ acquires lock \Lock if condition $\EE$ holds (and waits otherwise), executes \Command, and releases \Lock
again upon termination of $\Command$;
$\WithinBlock{\Lock}{\Command}$ is an internal command that we will use to indicate that $\Command$ is executed while holding the lock $\Lock$.
Furthermore, \Deref{\EE} denotes the value at the memory address given by \EE.
$\Lookup{\Var}{\EE}$ \emph{reads} the value at address \EE and assigns it to \Var;
\Mutation{\EE}{\EE'} \emph{writes} the value of $\EE'$ to the address \EE.
Moreover, \Free{\EE} disposes of location \EE, and $\New{\Var}{\EE}$ allocates a free memory address, assigns it to \Var, and stores \EE at that address.
We use structured concurrency and global locks to simplify the formalization, but our methodology also supports dynamic threads and locks.

\subsection{Program States}
\label{sec:sound:states}
A \emph{program state} $(\Stack,\PHeap)$ consists of a \emph{stack} \Stack, \ie, a valuation
of variables, and a \emph{heap} \PHeap modeling dynamically allocated memory.
Formally, we fix a set $\Vals = \Set{\Val,\ldots}$ of \emph{values} containing, \eg, Booleans, integers, and sequences. The set \Stacks consists of all mappings \Stack from variables in $\Vars$
to values in $\Vals$, \ie,
\[ \Stacks \ddefeq \Vars \to \Vals~. \]
Moreover, we fix a countably infinite set $\Locs$ of \emph{memory addresses}, and the set $\PermsSym \defeq \Perms$ of \emph{fractional permissions}, where permission $\Perm = \PermWrite$ means write access, $\Perm \in (0,1)$ means read access, and $\Perm = \PermNone$ means no access, respectively.

The set \PHeaps of \emph{heaps} consists of all finite partial functions \PHeap that map addresses in their domain $\dom{\PHeap} \subseteq \Locs$ to permission-value pairs, \ie,
\begin{align*}
  \PHeaps \ddefeq \Locs \fpto (\PermsSym \times \Vals)~.
\end{align*}
For $\PHeap(\Loc) = (\Perm,\Val)$, we define the projections
$\HPerm{\PHeap(\Loc)} = \Perm$ and $\HVal{\PHeap(\Loc)} = \Val$.
We denote by $\Set{\Loc_1 \xmapsto{\Perm_1} \Val_1, \ldots, \Loc_n \xmapsto{\Perm_n} \Val_n}$ the heap $\PHeap$ with $\dom{\PHeap} = \Set{\Loc_1,\ldots,\Loc_n}$ and $\PHeap(\Loc_i) = (\Perm_i,\Val_i)$, where $i \in \Set{1,\ldots,n}$; the empty heap is denoted by $\EmptyHeap$.

The \emph{addition} of $\PHeap_1 \Add \PHeap_2$ of heaps $\PHeap_1$ and $\PHeap_2$ is defined as
\begin{align*}
  (\PHeap_1 \Add \PHeap_2)(\Loc) \ddefeq
  \begin{cases}
    (\Perm_1 + \Perm_2,\Val) & \tif \Loc \in \dom{\PHeap_1} \cap \dom{\PHeap_2} \tand \PHeap_1(\Loc) = (\Perm_1,\Val)  \\
                 & \qquad\qquad \tand \PHeap_2(\Loc) = (\Perm_2,\Val) \tand \Perm_1 + \Perm_2 \leq \PermWrite \\
   \PHeap_1(\Loc) & \tif \Loc \in \dom{\PHeap_1} \setminus \dom{\PHeap_2} \\
   \PHeap_2(\Loc) & \tif \Loc \in \dom{\PHeap_2} \setminus \dom{\PHeap_1} \\
   \undefined & \text{otherwise}. 
  \end{cases}
\end{align*}
We write $\isDefined{\PHeap_1 \Add \PHeap_2}$ if ($\PHeap_1 \Add \PHeap_2(a)$ is well-defined for all $a \in \dom{\PHeap_1} \cup \dom{\PHeap_2}$.
Furthermore,
$\PHeap\Update{\Loc}{\Val}$ denotes the heap $\PHeap$ except that $\Loc$ maps to $\Val$ (with permission $\PermWrite$), \ie,
\begin{align*}
  \PHeap\Update{\Loc}{\Val}(\Loc') \ddefeq
  \begin{cases}
    (\PermWrite, \Val) & \tif \Loc' = \Loc \\
    \PHeap(\Loc') & \tif \Loc' \neq \Loc~.
  \end{cases}
\end{align*}
For every stack $\Stack$, $\Stack\Update{\Var}{\Val}$ is defined analogously.
We define the heap
$\PHeap\Delete{\Loc} \defeq \Set{\Loc' \mapsto \PHeap(\Loc') \mid \Loc' \in \dom{\PHeap} \setminus \{\Loc\}}$, which is obtained from $\PHeap$ by removing $\Loc$ from its domain.
Finally, the set \Heaps of \emph{normal heaps} consists of all heaps $\Heap$ that assign the  full permission $\PermWrite$ to every address, \ie,
\begin{align*}
  \Heaps \ddefeq \Set{ \PHeap \in \PHeaps ~|~ \forall \Loc \in \dom{\PHeap} \mydot \HPerm{\PHeap(\Loc)} = \PermWrite }~.
\end{align*}
We typically write $\Heap$ instead of $\PHeap$ to highlight that a heap is normal; by slight abuse of notation, we use $\Heap(\Loc)$ as a shortcut for $\HVal{\Heap(\Loc)}$.
Clearly, for every heap $\PHeap$, there exists a heap $\PHeap'$ such that $\PHeap \Add \PHeap'$ is a normal heap.

\subsection{Operational Semantics}
\label{sec:sound:semantics}
The semantics of commands is defined in terms of
a small-step execution relation.

Toward a formal definition, we first clarify the semantics of expressions.
We do not fix a specific syntax for expressions; instead, we assume that every expression
 $\EE$ is associated with a function 
$\EE\colon \Stacks \to \Vals$
such that $\EE(\Stack)$ is the value obtained from evaluating $\EE$ in stack $\Stack$.

The set $\Confs$ of \emph{program configurations} \Configuration consists of all 
command-stack-normal-heap triples plus a dedicated error state \Abort, that is,
\begin{align*}
  \Confs \ddefeq (\Commands \times \Stacks \times \Heaps) ~\cup~ \Set{\Abort}~.
\end{align*}
We use normal heaps because permissions are a reasoning concept that does not exist when executing a program.
The \emph{small-step operational semantics} of commands is defined as the execution relation $\OpRel \subseteq \Confs \times \Confs$ given by the rules in \figref{fig:operational}.
\begin{figure}[tp]
\centering
\input{fig-prelim-operational}
\caption{Rules of the small-step operational semantics for commands in $\Commands$.
Here, $\Locked{\Command}$ is the set of all locks $\Lock$ held by $\Command$, \ie,
$\WithinBlock{\Lock}{\ldots}$ is a sub-command of $\Command$.
Furthermore, $\Reads{\Command}{\Stack}$ (resp. $\Writes{\Command}{\Stack}$) denotes the 
set of those addresses $\Loc$ that are not exclusively owned by $\Command$, \ie, appear
only in sub-commands $\NextBlock{\ldots}$ or $\WithBlock{\Lock}{\EE}{\ldots}$), and 
are accessed (resp. modified) by \Command given stack $\Stack$.
}
\label{fig:operational}
\end{figure}
Most rules in \figref{fig:operational} are standard (\cf~\cite{Vafeiadis11}) and reflect the behavior informally explained in \secref{sec:sound:language}; we briefly discuss particularities.

\paragraph{Usage and Declaration of Locks.}
The operational semantics uses the command structure to record which locks are currently declared and acquired, respectively.
That is, $\DeclareLockBlock{\Lock}{\CC}$ indicates that \Lock is declared in \CC; the internal command $\WithinBlock{\Lock}{\CC}$ indicates that \CC holds $\Lock$ during its execution.
We denote by $\DLocks{\CC}$ the set of all locks that are declared in \CC.
Analogously, \Locked{\CC} denotes the set of all locks that are currently held by \CC.
For example, if $\CC'$ is \emph{not} of the form $\DeclareLockBlock{\Lock}{\ldots}$ or $\WithinBlock{\Lock}{\ldots}$, and 
$\CC$ is given by
\begin{align*}
  \DeclareLockBlock{\Lock_1}{
    \DeclareLockBlock{\Lock_2}{
        \ldots
        \DeclareLockBlock{\Lock_n}{
            \WithinBlock{\Lock_1'}{
                \ldots
                \WithinBlock{\Lock_m'}{
                    \CC'
                }
            }
        }
    } \ldots
  }~,
\end{align*}
then $\DLocks{\CC} = \Set{ \Lock_1, \ldots, \Lock_n}$
and $\Locked{\CC} = \Set{ \Lock_1', \ldots, \Lock_m' }$.

The semantics of $\DeclareLockBlock{\Lock}{\CC}$ consequently keeps the lock $\Lock$ declared and performs a step of $\CC$ until termination; after that, the lock declaration expires.

\paragraph{CCRs and Parallel Composition.}
For CCRs $\WithBlock{\Lock}{\EE}{\CC}$, we first check whether condition \EE holds in the current state. 
If so, we acquire lock $\Lock$ and enter the CCR by moving to the internal command $\WithinBlock{\Lock}{\CC}$, which will execute \CC, and release the lock upon termination of \CC.
If \EE does not hold, no transition is possible, \ie, we need to wait for other threads.

The rules for the parallel composition $\CC_1 \Par \CC_2$ model all interleaved executions of $\CC_1$ and $\CC_2$.
They include a sanity check stating that $\CC_1$ and $\CC_2$ do not hold the same lock at the same time, \ie, locks provide mutual exclusion.
Moreover, there is a rule that prematurely aborts executions whenever a command admits a data race, that is, $\CC_1$ (resp.\ $\CC_2$) writes to and $\CC_2$ (resp.\ $\CC_1$) accesses (i.e., reads from or writes to) the same address that is \emph{outside} of a CCR, and thus without protection by a lock.

\subsection{Assertions}
\label{sec:sound:assertions}

\paragraph{Syntax.}
The syntax of \emph{separation logic assertions}
includes, amongst others,
all Boolean expressions $\EE$ supported by our programming language
Formally, the set $\SL$ of separation logic assertions is given by the grammar
\begin{align*}
    \PP \cceq
        \EE ~|~ &  \PP \wedge \PP ~|~ 
\neg \PP ~|~ 
         \fa{\Var} \PP ~|~ \ex{\Var} \PP \tag{FOL} \\
    |~ & \emp ~|~ \pt{\EE}{\Perm}{\EE}
   ~|~  \PP \sep \PP ~|~ \PP \mw \PP~ ~|~ \bigstar_{i \in I}\,\PP_i~, 
\tag{SL}
\end{align*}
where $\EE$ is a Boolean expression over variables,
$\Var$ is a variable in $\Vars$,
$\Perm$ is a permission in $\PermsSym$,
and $I$ is a finite set such that each $i \in I$ is associated with an assertion $\PP_i$.
We use syntactic sugar, such as
$\PP \Implies \PQ$, $\PP \vee \PP$ and, in particular,
$\apt{\EE}{\Perm}{\EE} \defeq \pt{\EE}{\Perm}{\EE} \sep \True$, $\alloc{\EE}{\Perm} \defeq \ex{y} \apt{\EE}{\Perm}{y}$, and
$\acc{\EE}{\Perm} \defeq \ex{y} \apt{\EE}{\Perm}{y}$.

We denote by $\FOL$ the set of all first-order logic formulas, \ie, those
formulas that can be constructed from the first line of the above grammar.

\paragraph{Semantics.}
Assertions are interpreted over pairs $(\Stack,\PHeap)$ consisting of a stack \Stack and a heap (with permissions) \PHeap.
\tabref{tab:assertion-semantics} shows the formal semantics of assertions.

Intuitively, $\emp$ specifies the empty heap;
the points-to assertion $\pt{\EE}{\Perm}{\EE'}$ specifies that the heap contains 
\emph{exactly one} address $\EE$ that is mapped to $\EE'$ with permission $\Perm$; 
$\apt{\EE}{\Perm}{\EE'}$ is an intuitionistic version stating that the heap contains 
\emph{at least} permission $\Perm$ for the address $\EE$, which is mapped to $\EE'$.
$\alloc{\EE}{\Perm}$ and $\acc{\EE}{\Perm}$ are analogous but do not require a specific
value at the address given by $\EE$.
The separating conjunction $\PP \sep \PQ$ specifies that the heap can be partitioned into two parts such that one part satisfies \PP and the other part satisfies \PQ.
$\bigstar_{i \in I}\,\PP_i$ is an iterative version of the separating conjunction.
Finally, the magic wand $\PP \mw \PQ$ specifies that \PQ holds in a heap \PHeap after it has been extended by any heap that satisfies \PP (and can be added to \PHeap).

We call an assertion $\PP$ \emph{valid}, written $\valid \PP$, if and only if
$\Stack,\PHeap \models \PP$ holds for all stacks $\Stack \in \Stacks$ and heaps $\PHeap \in \PHeaps$.

{\renewcommand{\arraystretch}{1.2}
\begin{table}[t]
\centering
\caption{Semantics of assertions.}
\label{tab:assertion-semantics}
\begin{tabular}{l@{\hskip 0.25in}l}
  $\PP$ & $\Stack, \PHeap \models \PP \qiff$ \\
  \hline\hline
  $\EE$ & \Sem{\EE}(s) = \True \\
  $\PQ \wedge \PR$ & $\Stack, \PHeap \models \PQ$ and $\Stack, \PHeap \models \PR$ \\
  $\neg\PQ$ & $\Stack, \PHeap \not\models \PQ$ \\
  $\fa{\Var} \PQ$ & for all $\Val \in \Vals$, $\Stack\Update{\Var}{\Val},\PHeap \models \PQ$ \\
  $\ex{\Var} \PQ$ & exists $\Val \in \Vals$ s.t. $\Stack\Update{\Var}{\Val},\PHeap \models \PQ$ \\
  $\emp$ & $\dom{\PHeap} = \emptyset$ \\
  $\PQ \sep \PR$ & exists $\PHeap_1,\PHeap_2$ s.t. $\PHeap = \PHeap_1 \PermAdd \PHeap_2$ 
  and $\Stack,\PHeap_1 \models \PQ$ and $\Stack,\PHeap_2 \models \PR$ \\
  $\PQ \mw \PR$ & for all $\PHeap'$, ($\isDefined{\PHeap \Add \PHeap'}$ and $\Stack, \PHeap' \models \PQ$) implies
    $\Stack,\PHeap \PermAdd \PHeap' \models \PR$  \\
    $\bigstar_{i \in I}\,\PP_i$ & $I = \emptyset$ or exists $j \in I$ s.t. $\Stack,\PHeap \models \PP_j \sep \bigstar_{i \in I \setminus \{j\}}\,\PP_i$
\\
  \hline \hline
\end{tabular}
\end{table}
}

We denote by $\FV{\PP}$ the set of \emph{free variables} (\ie, those that are not bound by quantifiers) of assertion $\PP$.
Moreover, $\PP\subst{\Var}{\Val}$ denotes the \emph{substitution} of every free occurrence of 
variable $\Var$ in assertion $\PP$ by $\Val$.

\subsection{Proof System}
\label{sec:sound:logic}
The last step of our recapitulation of concurrent separation logic (CSL)
presents the formal triples and proof rules for reasoning about commands.
To formalize both, we first introduce lock environments and lock invariants.

\paragraph{Lock Environments.}
In contrast to the operational semantics (\cf~\secref{sec:sound:semantics}), which has a global view on the full program state, the CSL proof system considers only the local state \ie, those parts of the current program state that are accessible to the local command.
The remainder of the global state is either framed around the command or protected by locks.
To this end, we associate every declared lock $\Lock$ with a \emph{lock invariant} $\PR$ that specifies the part of the global state that is protected by $\Lock$.
Formally, the assignment of lock invariants to locks is 
captured by a \emph{lock environment} $\LockEnv$, which is given by the grammar
\begin{align*}
    \LockEnv \cceq &
    \emptyset ~|~ \LockEnv, \LockDef{\Lock}{\PR}~,
\end{align*}
where $\Lock \in \Locks$ and $\PR$ is an $\SL$ assertion representing the lock invariant.
Moreover, given a lock environment $\LockEnv$,
the \emph{lock invariant of lock $\Lock$} is 
\begin{align*}
 \LockEnv(\Lock) \ddefeq
 \begin{cases}
    \PR & \tif \LockEnv = \LockEnv', \LockDef{\Lock}{\PR} \\
    \LockEnv'(\Lock) & \tif \LockEnv = \LockEnv', \LockDef{\Lock'}{\PR'} \tand \Lock' \neq \Lock  \\
    \emp & \tif \LockEnv = \emptyset~.
 \end{cases}
\end{align*}
Intuitively, if $\LockEnv(\Lock) = \PR$, then $\PR$ specifies a portion of the global state that is not part of the local state but is shared with other threads and can be modified by them.
In particular, if $\LockEnv(\Lock) = \emp$, then $\Lock$ does not appear in $\LockEnv$ (or the declared invariant is $\emp$) and nothing is shared with other threads.

Declaring, acquiring, and releasing the lock $\Lock$ can be understood as a transfer of the portion of the state specified by $\PR$ between the local state and \LockEnv, \ie, the part of the global state that is shared with other threads:
whenever we declare \Lock, \PR is shared with other threads and thus moved from the local state into \LockEnv; once the declaration expires, \PR is transferred back into the local state. 
Whenever we acquire the lock $\Lock$, \PR is moved from $\LockEnv$ into the local state; whenever we release the lock $\Lock$, $\PR$ is moved back from the local state into $\LockEnv$.

\paragraph{CSL Judgments.}
Judgments that can be derived in CSL are of the form
\begin{align*}
  \striple{\PP}{\CC}{\PQ}~,
\end{align*}
where
\LockEnv is a lock environment such that $\bigstar_{\Lock \in \Locks} \LockEnv(\Lock)$
describes the shared state,
$\PP \in \SL$ is the precondition evaluated in the local state,
\CC is a command in \Commands, and
$\PQ \in \SL$ is the postcondition evaluated in the local state.
\figref{fig:proof-rules} shows the standard proof rules of CSL~\cite{Vafeiadis11}.
In particular, the aforementioned transfer between local and shared state can be observed in the rules \Rule{Lock} and \Rule{With}.
A detailed discussion of the CSL proof rules is found in~\cite{OHearn04,Brookes07}. 
\begin{figure}[p]
  \centering
  \input{fig-proof-rules}
  \caption{Inference rules inherited from concurrent separation logic. 
           Here, $\Perm$ is a permission; $\Mod{\Command}$ denotes all variables
           modified by command \CC, \ie, those that appear on the left-hand side of 
           assignments. Moreover, $\FV{\EE}$, $\FV{\CC}$, and $\FV{\LockEnv}$ denote the 
           free variables of in expression $\EE$, command \CC (\ie, all variables accessed by \CC),
           and lock environment \LockEnv (\ie, the union of $\FV{\LockEnv(\Lock)}$
           for all $\Lock \in \Locks$), respectively.
           Furthermore,  we define
           $\FV{A,B,C} \defeq \FV{A} \cup \FV{B} \cup \FV{C}$.
    Finally, an assertion $\PP$ is \emph{precise} iff 
    for all $\Stack$, $\PHeap_1,\PHeap_2,\PHeap_1',\PHeap_2'$,
    if \isDefined{\PHeap_1 \Add \PHeap_2} and
    $\PHeap_1 \Add \PHeap_2 = \PHeap_1' \Add \PHeap_2'$
    and $\Stack,\PHeap_1 \models \PP$ and
    and $\Stack,\PHeap_2 \models \PP$, then $\PHeap_1 = \PHeap_2$.
  }
  \label{fig:proof-rules}
\end{figure}

%% file: fig-prelim-operational.tex
\begin{tabular}{l@{\hskip 0.05in}l@{\hskip 0.05in}l}
\Rule{Assign}
&
$\Conf{\Assign{\Var}{\EE}}{\Stack}{\Heap}
\OpRel
\Conf{\Skip}{\Stack\Update{\Var}{\Sem{\EE}(\Stack)}}{\Heap}$
&
\\
\Rule{Read}
&
$\Conf{\Lookup{\Var}{\EE}}{\Stack}{\Heap}
\OpRel
\Conf{\Skip}{\Stack\Update{\Var}{\Heap(\Val)}}{\Heap}$
&
if $\Sem{\EE}(\Stack) = \Val \in \dom{\Heap}$
\\
\Rule{ReadA}
&
$\Conf{\Lookup{\Var}{\EE}}{\Stack}{\Heap} \OpRel \Abort$
&
if $\Sem{\EE}(\Stack) \notin \dom{\Heap}$
\\
\Rule{Write}
&
$\Conf{\Mutation{\EE}{\EE'}}{\Stack}{\Heap}
 \OpRel
 \Conf{\Skip}{\Stack}{\Heap\Update{\Loc}{\Sem{\EE'}(\Stack)}}$
&
if $\Sem{\EE}(\Stack) = \Loc \in \dom{\Heap}$
\\
\Rule{WriteA}
&
$\Conf{\Mutation{\EE}{\EE'}}{\Stack}{\Heap} \OpRel \Abort$
&
if $\Sem{\EE}(\Stack) \notin \dom{\Heap}$
\\
\Rule{Alloc}
&
$\Conf{\New{\Var}{\EE}}{\Stack}{\Heap}
 \OpRel
 \Conf{\Skip}{\Stack\Update{\Var}{\Loc}}{\Heap\Update{\Loc}{\Sem{\EE}(\Stack)}}$
&
if $\Loc \in \Locs \setminus \dom{\Heap}$ 
\\
\Rule{Free}
&
$\Conf{\Free{\EE}}{\Stack}{\Heap}
 \OpRel
 \Conf{\Skip}{\Stack}{\Heap\Delete{\Loc}}$
&
if $\Sem{\EE}(\Stack) = \Loc \in \dom{\Heap}$ 
\\
\Rule{FreeA}
&
$\Conf{\Free{\EE}}{\Stack}{\Heap}
 \OpRel \Abort$
&
if $\Sem{\EE}(\Stack) \notin \dom{\Heap}$
\\
\Rule{With}
&
$\Conf{\WithBlock{\Lock}{\EE}{\Command}}{\Stack}{\Heap}
 \OpRel
 \Conf{\WithinBlock{\Lock}{\Command}}{\Stack}{\Heap}$
&
if $\EE(\Stack) = \True$
\\
\Rule{Ite1}
&
$\Conf{\Ite{\BB}{\Command_1}{\Command_2}}{\Stack}{\Heap}
 \OpRel
 \Conf{\Command_1}{\Stack}{\Heap}$
&
if $\Sem{\BB}(\Stack) = \True$
\\
\Rule{Ite2}
&
$\Conf{\Ite{\BB}{\Command_1}{\Command_2}}{\Stack}{\Heap}
 \OpRel
 \Conf{\Command_2}{\Stack}{\Heap}$
&
if $\Sem{\BB}(\Stack) = \False$
\\
\Rule{While}
&
\multicolumn{2}{l}{
$\Conf{\While{\BB}\Block{\Command}}{\Stack}{\Heap}
 \OpRel
 \Conf{
   \Ite{\BB}{\Command\Seq\While{\BB}\Block{\Command}}{\Skip}
 }{\Stack}{\Heap}$
}
\\
\hline
\Rule{WithinL}
&
$\Conf{\WithinBlock{\Lock}{\Command}}{\Stack}{\Heap} \OpRel \Abort$
&
if $\Lock \in \Locked{\Command}$
\\
\Rule{WithinS}
&
\multicolumn{2}{l}{
  $\Conf{\WithinBlock{\Lock}{\Skip}}{\Stack}{\Heap} \OpRel \Conf{\Skip}{\Stack}{\Heap}$
  \quad
  \Rule{SeqS}~
  $\Conf{\Skip\Seq\Command_2}{\Stack}{\Heap}  \OpRel \Conf{\Command_2}{\Stack}{\Heap}$
}
\\
\Rule{LockS}
&
\multicolumn{2}{l}{
  $\Conf{\DeclareLockBlock{\Lock}{\Skip}}{\Stack}{\Heap} \OpRel \Conf{\Skip}{\Stack}{\Heap}$
  \quad
  \Rule{ParS}
  $\Conf{\Skip\Par\Skip}{\Stack}{\Heap} \OpRel \Conf{\Skip}{\Stack}{\Heap}$
}
\end{tabular}

\vspace{1em}

\begin{adjustbox}{center}
\begin{tabular}{c}
\inference[\Rule{Seq}]{
  \Conf{\Command_1}{\Stack}{\Heap}
  \OpRel
  \Conf{\Command_1'}{\Stack'}{\Heap'}
}{
  \Conf{\Command_1\Seq\Command_2}{\Stack}{\Heap}
  \OpRel
  \Conf{\Command_1'\Seq\Command_2}{\Stack'}{\Heap'}
}
\qquad
\inference[\Rule{SeqA}]{
  \Conf{\Command_1}{\Stack}{\Heap}
  \OpRel
  \Abort
}{
  \Conf{\Command_1\Seq\Command_2}{\Stack}{\Heap}
  \OpRel
  \Abort
}
\\[1.5em]
\inference[\Rule{Par1}]{
  \Conf{\Command_1}{\Stack}{\Heap}
  \OpRel
  \Conf{\Command_1'}{\Stack'}{\Heap'}
  \\
  \Locked{\Command_1'} \cap \Locked{\Command_2} \eeq \emptyset
}{
  \Conf{\Command_1\Par\Command_2}{\Stack}{\Heap}
  \OpRel
  \Conf{\Command_1'\Par\Command_2}{\Stack'}{\Heap'}
}
\quad
\inference[\Rule{Par1A}]{
  \Conf{\Command_1}{\Stack}{\Heap}
  \OpRel
  \Abort
}{
  \Conf{\Command_1\Par\Command_2}{\Stack}{\Heap}
  \OpRel
  \Abort
}
\\[1.5em]
\inference[\Rule{Par2}]{
  \Conf{\Command_2}{\Stack}{\Heap}
  \OpRel
  \Conf{\Command_2'}{\Stack'}{\Heap'}
  \\
  \Locked{\Command_1} \cap \Locked{\Command_2'} \eeq \emptyset
}{
  \Conf{\Command_1\Par\Command_2}{\Stack}{\Heap}
  \OpRel
  \Conf{\Command_1'\Par\Command_2}{\Stack'}{\Heap'}
}
\quad
\inference[\Rule{Par2A}]{
  \Conf{\Command_2}{\Stack}{\Heap}
  \OpRel
  \Abort
}{
  \Conf{\Command_1\Par\Command_2}{\Stack}{\Heap}
  \OpRel
  \Abort
}
\\[1.5em]
\inference[\Rule{Race}]{
    \left(\Reads{\Command_1}{\Stack} \cap \Writes{\Command_2}{\Stack}\right)
    \,\cup\,
    \left(\Writes{\Command_1}{\Stack} \cap \Reads{\Command_2}{\Stack}\right) \neq \emptyset
}{
  \Conf{\Command_1\Par\Command_2}{\Stack}{\Heap}
  \OpRel
  \Abort
}
\\
\inference[\Rule{Lock}]{
  \Conf{\Command}{\Stack}{\Heap} \OpRel \Conf{\Command'}{\Stack'}{\Heap'}
}{
  \Conf{\DeclareLockBlock{\Lock}{\Command}}{\Stack}{\Heap}
  \OpRel
  \Conf{\DeclareLockBlock{\Lock}{\Command'}}{\Stack'}{\Heap'}
}
\quad
\inference[\Rule{LockA}]{
  \Conf{\Command}{\Stack}{\Heap} \OpRel \Abort
}{
  \Conf{\DeclareLock{\Lock}\Block{\Command}}{\Stack}{\Heap}
  \OpRel
  \Abort
}
\\[1.5em]
\inference[\Rule{Within}]{
  \Conf{\Command}{\Stack}{\Heap}
  \OpRel
  \Conf{\Command'}{\Stack'}{\Heap'}
}{
  \Conf{\WithinBlock{\Lock}{\Command}}{\Stack}{\Heap}
  \OpRel
  \Conf{\WithinBlock{\Lock}{\Command'}}{\Stack'}{\Heap'}
}
\\[1.5em]
\inference[\Rule{WithinA}]{
  \Conf{\Command}{\Stack}{\Heap} \OpRel \Abort
}{
  \Conf{\WithinBlock{\Lock}{\Command}}{\Stack}{\Heap}
  \OpRel
  \Abort
}
\end{tabular}
\end{adjustbox}

%% file: fig-proof-rules.tex
\begin{tabular}{c}
\inference[\Rule{Skip}]{}{
  \striple{\PP}{\Skip}{\PP}
}
{\hskip 1.5cm}
\inference[\Rule{Assign}]{
 \Var \notin \FV{\LockEnv}
}{
  \striple{
      \PP\subst{\Var}{\EE}
  }{
      \Assign{\Var}{\EE}
  }{\PP}
}
\\[1.5em]
\inference[\Rule{Write}]{}{
    \striple{
        \alloc{\EE}{\PermWrite}
    }{
        \Mutation{\EE}{\EE'}
    }{
        \pt{\EE}{\PermWrite}{\EE'} 
    }
}
\\[2em]
\inference[\Rule{Read}]{
    \Var \notin \FV{\EE,\EE',\LockEnv}
}{
    \striple{
        \pt{\EE}{\Perm}{\EE'}
    }{
        \Lookup{\Var}{\EE}
    }{\pt{\EE}{\Perm}{\EE'} \wedge \Var = \EE'}
}
\\[1.5em]
\inference[\Rule{Alloc}]{
    \Var \notin \FV{\LockEnv,\EE}
}{
    \striple{\emp}{
        \New{\Var}{\EE}
    }{
        \pt{\Var}{\PermWrite}{\EE} 
    }
}
\\[1.5em]
\inference[\Rule{Free}]{
    \EE \notin \ALocs
}{
    \striple{
        \alloc{\EE}{\PermWrite}
    }{
        \Free{\EE}
    }{
        \emp
    }
}
\\[2em]
\inference[\Rule{Seq}]{
    \striple{\PP}{\Command_1}{\PR}
    &
    \striple{\PR}{\Command_2}{\PQ}
}{
    \striple{
        \PP
    }{
        \Command_1\Seq\Command_2
    }{
        \PQ
    }
}
\\[1.5em]
\inference[\Rule{Cond}]{
    \striple{\PP \wedge \EE}{\Command_1}{\PQ}
    &
    \striple{\PP \wedge \neg\EE}{\Command_2}{\PQ}
}{
    \striple{
        \PP
    }{
        \If{\EE}\Block{\Command_1}\Else\Block{\Command_2}
    }{
        \PQ
    }
}
\\[1.5em]
\inference[\Rule{While}]{
    \striple{\PI \wedge \EE}{\Command}{\PI}
}{
    \striple{
        \PI
    }{
        \While{\EE}\Block{\Command}
    }{
        \PI \wedge \neg\EE
    }
}
\\[1.5em]
\inference[\Rule{Par}]{
    \striple{\PP_1}{\Command_1}{\PQ_1}
    &
    \FV{\PP_1, \Command_1, \PQ_1} \cap \Mod{\Command_2} = \emptyset
    \\
    \striple{\PP_1}{\Command_2}{\PQ_2}
    &
    \FV{\PP_2, \Command_2, \PQ_2} \cap \Mod{\Command_1} = \emptyset
}{
    \striple{
        \PP_1 \sep \PP_2
    }{
        \Command_1\Par\Command_2
    }{
        \PQ_1 \sep \PQ_2
    }
}
\\[1.5em]
\inference[\Rule{Lock}]{
  \triple{}{\LockEnv,\LockDef{\Lock}{\PR}}{
    \PP
  }{
    \Command
  }{
    \PQ
  }
}{
  \triple{}{\LockEnv}{
    \PR \sep \PP
  }{
      \DeclareLock{\Lock}\Block{\Command}
  }{
      \PR \sep \PQ 
  }
}
\\[1.5em]
\inference[\Rule{With}]{
  \striple{(\PP \sep \PR) \wedge \EE}{\Command}{\PQ \sep \PR}
}{
  \triple{}{\LockEnv,\LockDef{\Lock}{\PR}}
    {\PP}{\WithBlock{\Lock}{\EE}{\Command}}{\PQ}
}
\\[1.5em]
\inference[\Rule{Frame}]{
  \striple{\PP}{\Command}{\PQ} 
  &
  \FV{\PR} \cap \Mod{\Command} = \emptyset
}{
  \striple{\PP \sep \PR}{\Command}{\PQ \sep \PR}
}
\\[1.5em]
\inference[\Rule{Cons}]{
  \striple{\PP}{\Command}{\PQ} 
  \\
  \valid \PP' \Implies \PP
  \qquad \valid
  \PQ \Implies \PQ'
}{
  \striple{\PP'}{\Command}{\PQ'}
}
\qquad
\inference[\Rule{Ex}]{
  \striple{\PP}{\Command}{\PQ} 
  &
  \Var \notin \FV{\Command}
}{
  \striple{\ex{\Var} \PP}{\Command}{\PQ}
}
\\[1.5em]
\inference[\Rule{Conj}]{
  \forall \Lock \mydot \LockEnv(\Lock)~\text{precise}
  \\
  \striple{\PP_1}{\Command}{\PQ_1}
  \\
  \striple{\PP_2}{\Command}{\PQ_2}
}{
  \striple{\PP_1 \wedge \PP_2}{\Command}{\PQ_1 \wedge \PQ_2}
}
\qquad
\inference[\Rule{Disj}]{
  \striple{\PP_1}{\Command}{\PQ_1}
  \\
  \striple{\PP_2}{\Command}{\PQ_2}
}{
  \striple{\PP_1 \vee \PP_2}{\Command}{\PQ_1 \vee \PQ_2}
}
\end{tabular}

%% file: methodology.tex
We now present the details of our methodology and formalize it on
top of the concurrent separation logic (CSL) introduced in
\secref{sec:preliminaries}.
We first discuss how abstract models are encoded into program states.
After that, we extend the programming language, its operational semantics, and the CSL proof system with refinement-specific commands and rules.
Finally, we show that our methodology is sound, \ie, a proof in our program logic
guarantees that the (finite) traces of the given implementation are included in the traces of the abstract model. Further details and proofs are provided in the appendix.

\subsection{Abstract Models}
Our methodology takes as input an abstract model that should be refined by a concrete implementation.
More precisely, we assume that an abstract model is provided as a
(potentially infinite-state) abstract transition system.

Toward a formal definition, recall from \secref{sec:sound:assertions} the definition 
of heap-independent assertions in first-order logic (\FOL).
Moreover, given a sequence of variables $\vec{x} = (x_1, \ldots, x_k)$ , we denote by
$\Primed{\vec{x}}$ the same sequence in which each $x_i$ is replaced 
by a primed version, \ie, $\Primed{\vec{x}} = (\Primed{x_1}, \ldots, \Primed{x_k})$.
\begin{definition}[Abstract Transition System (ATS)]\label{def:ts}
An abstract transition system is a quadruple
$\AbsTS \defeq (\AbsLength,\AbsVars, \AbsInit, \AbsNext)$, where
\begin{itemize}
  \item $\AbsVars = (\AbsVar{1}, \ldots, \AbsVar{\AbsLength})$ is a
         repetition-free sequence of $\AbsLength \geq 1$ variables,
  \item the \emph{initial state formula} $\AbsInit(\AbsVars)$ is an $\FOL$ assertion over $\AbsVars$, and
  \item the \emph{next state formula} $\AbsNext(\AbsVars, \Primed{\AbsVars})$ is an $\FOL$ assertion over $\AbsVars$ and $\Primed{\AbsVars}$.
\EndDef
\end{itemize}
\end{definition}
We consider only ATSs that are \emph{stutter-invariant},
that is, for all sequences of values $\vec{\Val} \in \Vals^{\AbsLength}$, the formula
$\AbsNext(\vec{\Val},\vec{\Val})$ is valid.
Stutter invariance is desirable in the context of refinement proofs as concrete implementations
should always be allowed to perform more fine-grained computation steps.
Furthermore, every ATS can be turned into a stutter-invariant one by considering the modified next-state
formula 
$
   \AbsNext(\AbsVars,\Primed{\AbsVars}) ~\vee~ \bigwedge_{1 \leq i \leq \AbsLength} \Primed{\AbsVar{i}} = \AbsVar{i}$.

Throughout the rest of this paper, we fix a stutter-invariant
abstract transition system 
$\AbsTS \defeq (\AbsLength,\AbsVars, \AbsInit, \AbsNext)$ representing
the abstract model we would like to prove refinement of.

To reason about refinement, we need to formalize the observable traces
 of $\AbsTS$.
Every evaluation of an \AbsTS's variables $\AbsVars$ 
constitutes one of its \emph{states}.
Formally, the \emph{state space} of $\AbsTS$ is $\AbsStates \defeq \Vals^{\AbsLength}$.
A (finite) \emph{path} of $\AbsTS$ is a sequence of states $\AbsState_1 \ldots \AbsState_n$, where $n \geq 1$, such that
\begin{enumerate}
\item $\valid \AbsInit(\AbsState_1)$, \ie, $\AbsState_1$ is an initial state of \AbsTS, and 
\item for all $i \in [1,n)$, we have $\valid \AbsNext(\AbsState_i,\AbsState_{i+1})$, \ie,
      for every state but the last one, $\AbsTS$ admits a transition to the next state on the path.
\end{enumerate}
We collect in the set $\AbsPaths$ all finite paths of the transition system $\AbsTS$.

A \emph{trace} then projects every state on a path to the values corresponding to those variables
that are observable, \eg, variables modeling I/O channels as outlined in \secref{sec:overview}.
For simplicity, we assume that exactly one variable, $\AbsVar{1}$, 
is observable; generalizing to arbitrary sets of observable variables is
straightforward.
We denote by $\Obs{\AbsState}$ the projection of state $\AbsState$ to the value assigned
to the observable variable $\AbsVar{1}$, that is, if $\AbsState = (\Val_1,\ldots,\Val_\AbsLength)$, then
$\Obs{\AbsState} = \Val_1$. 
\begin{definition}
The set $\AbsTraces$ of \emph{observable traces} of $\AbsTS$ is given by
\begin{align*}
  \AbsTraces \ddefeq \{
     \Obs{\AbsState_1} \ldots \Obs{\AbsState_n}
     ~|~
     & n \in \Nats \tand \\
     & \AbsState_1 \ldots \AbsState_n \in \AbsPaths
  \}. \tag*{\EndDef}
\end{align*}
\end{definition}
\paragraph{State Encoding.}
\label{sec:methodology:state}
To reason about the behavior of $\AbsTS$s with the machinery
offered by concurrent separation logic (\cf~\secref{sec:preliminaries}), 
we encode states of $\AbsTS$s within program states---more precisely: as part of the heap.

We fix a set $\ALocs \subseteq \Locs$ of dedicated \emph{ghost addresses} that cannot be allocated or disposed of---they thus need to be already allocated before program execution.
Every variable $\AbsVar{i}$ of \AbsTS is then encoded as a ghost address
$\AbsLoc{i} \in \ALocs$ at which we store the current evaluation of
$\AbsVar{i}$.
We use $\stdout$ as a synonym for $\AbsLoc{1}$---the address of the only observable variable $\AbsVar{1}$.

Given a heap $\PHeap$, the contents of the addresses $\AbsLocs$ in $\PHeap$ 
encode the current state of $\AbsTS$.
We introduce a function $\AbsLocsOfName\colon \PHeaps \to \AbsStates$
which extracts the transition system's state from the heap:
\begin{definition}\label{def:methodology:get-state}
The \emph{state extraction function} $\AbsLocsOfName\colon \PHeaps \to \AbsStates$ is
\begin{align*}
  \AbsLocsOf{\PHeap} \ddefeq
  \begin{cases}
    (\Val_1,\ldots,\Val_\AbsLength) & \tif \AbsLocs \in \dom{\PHeap} \tand \\
                                    & \quad \text{for all } i \in [1,\AbsLength], ~ \PHeap(\AbsLoc{i}) = (\Perm_i,\Val_i) \\
& \qquad\qquad\qquad\qquad \tand \Perm_i > \PermNone \\
    \undefined & \text{otherwise}~. $\EndDef$
  \end{cases}
\end{align*}
\end{definition}
While ghost addresses cannot be allocated or disposed, (ghost) commands can read and modify
their contents just like with ordinary addresses.
Hence, we can enrich a concrete implementation with ghost commands to model updates to
the state of the abstract model \AbsTS.
\paragraph{Ghost Locks.}
\label{sec:methodology:locks}
We employ a dedicated lock $\GhostLock$ to protect the ghost locations
$\AbsLocs$ representing the current state of the abstract model $\AbsTS$
from undesirable modifications. 
To this end, we require that any lock invariant $\GhostInv$ chosen for $\GhostLock$
contains \emph{some} permission to each of the locations $\AbsLocs$---thus preventing modification of their content without acquiring $\GhostLock$ first. Formally:
\begin{assumption}\label{ass:ghostlock}
For any lock invariant $\GhostInv$ associated with the ghost lock $\GhostLock$,
there exists a permission $\Perm > 0$ such that the following entailment is valid:
\begin{align*}
  \valid~~ \big( 
  \GhostInv 
  \quad\Rightarrow\quad
   \acc{\AbsLoc{1}}{\Perm} \sep \ldots \sep 
   \acc{\AbsLoc{\AbsLength}}{\Perm}
  \big)
  \tag*{\EndDef}
\end{align*}
\end{assumption}
Our methodology prescribes that \GhostLock is a \emph{ghost lock} such that it can be safely
erased at runtime and thus cannot block execution.
To guarantee the above property, our formalization treats $\GhostLock$
differently from other locks: it is a dedicated lock that is invisible to programmers and
thus can neither be declared nor locked by them.
Instead, it will be governed by the refinement-specific ghost commands $\Initialize$ and $\Next$ in a way such that
$\GhostLock$ is indeed a ghost lock.
\subsection{Programming Language and Operational Semantics}
\label{sec:methodology:lang}
We now extend the programming language from 
\secref{sec:sound:language} with 
an output operation $\Print{\EE}$
and
the refinement-specific commands $\InitializeBlock{\CC}$ and $\NextBlock{\CC}$.
Formally, we expand the grammar in 
\secref{sec:sound:language} as follows:
\[
  \Command \cceq
  \Print{\EE} 
  ~|~ \InitializeBlock{\CC} ~|~ \NextBlock{\CC}
  ~|~ \ldots 
\]
We use ghost addresses taken from $\ALocs$ to formalize
the semantics of each of the above commands; \figref{fig:operational:novel} summarizes how we formally extend the execution relation $\OpRel$ from \secref{sec:sound:semantics}.
We briefly go over the new rules:
\begin{figure}[t]
\centering
\input{fig-operational-novel}

\caption{Extension of the rules in \figref{fig:operational} that determine the operational semantics.
Here, $\OpRelTrans$ denotes the reflexive and transitive closure of execution relation $\OpRel$.
}
\label{fig:operational:novel}
\end{figure}

$\Print{\EE}$ appends the value of $\EE$ to the standard output stream, which we represent by a mathematical sequence stored at the ghost address $\stdout \in \ALocs$.
Print thus modifies the only observable variable of the abstract mode $\AbsTS$.
If $\stdout$ has not been allocated, we abort execution.

The command $\InitializeBlock{\CC}$ is, operationally speaking, a dedicated command for declaring the ghost lock $\GhostLock$.
In fact, the rule \Rule{Init} desugars it to an explicit lock declaration $\DeclareLockBlock{\GhostLock}{\CC}$.\footnote{While we do not allow programmers to explicitly use the ghost lock $\GhostLock$, the rules of our operational semantics can use it like any other lock.}
Conceptually, $\Initialize$ takes the important role of marking the system as initialized, \ie, after entering $\Initialize$, the concrete model must be in an observable state that matches the (observable part of) the abstract model's initial state.
The command \NextBlock{\Command} is, operationally speaking again, a dedicated command for the conditional critical region
that acquires the ghost lock $\GhostLock$, executes $\CC$, and releases $\GhostLock$, that is, it can be viewed as syntactic sugar for $\WithBlock{\GhostLock}{\True}{\Command}$.
However, to ensure that $\GhostLock$ is indeed a \emph{ghost} lock and can thus be safely erased at runtime, we additionally require that every command $\CC$ put into a $\NextBlock{\CC}$ block is \emph{atomic}.
That is, $\CC$ can be executed in a single step without interference from other threads.
Consequently, $\NextBlock{\CC}$ is atomic and thus executed in a single step, as reflected by the \Rule{Next} rule: if the execution of $\Command$ terminates in a configuration $\Configuration$, then
$\NextBlock{\CC}$ transitions to configuration $\Configuration$ in a single step.

We assume that I/O operations, such as $\Print{\EE}$, are atomic and that adding ghost code to an atomic command yields an atomic command.
Apart from these assumptions, our methodology does not prescribe how to ensure atomicity---be it through known atomic statements, a global locking strategy, or formal proofs of atomicity in a program logic (\cf~\secref{sec:overview}).

\paragraph{Observable Traces.}
To formalize program refinement, 
we need to define the observable traces 
induced by a program configuration.
Similar to the definition of traces of the abstract model \AbsTS,
we project every program configuration in an execution
$\Configuration_1 \ldots \Configuration_n$ to its observable state once the system has been initialized, which is recorded by the operational semantics in its command structure:
\begin{definition}
  A command $\Command$ is \emph{initialized}, written $\isInit{\Command}$, if and only if the ghost lock $\GhostLock$ is declared in $\Command$, \ie, $\GhostLock \in \DLocks{\Command}$.
  \EndDef
\end{definition}
In our formalization, only the content of the standard output stream, which is modified via $\Print{\EE}$ calls, is observable.
The \emph{observable state} $\Obs{\Configuration}$ of a configuration
$\Configuration$ is defined as the empty or singleton sequence
\begin{align*}
  \Obs{\Configuration} \ddefeq 
  \begin{cases}
     \left[ \Heap(\stdout) \right] & \tif \Configuration = (\Conf{\CC}{\Stack}{\Heap})
                           \tand \isInit{\CC} \\
     \left[ ~ \right] & \text{otherwise}.
  \end{cases}
\end{align*}
A \emph{trace} is then obtained from a (finite) execution 
by mapping configurations to their observable state and concatenating the resulting sequences:
\begin{definition}\label{def:cmd:traces}
The set \Traces{\InitialConfs} of finite \emph{traces} induced by a set 
$\InitialConfs \subseteq \Confs$ is
\begin{align*}
  \Traces{\InitialConfs} \ddefeq \{\,
    \Obs{\Configuration_1} \ldots \Obs{\Configuration_n}
    \mid\, &
    n \in \Nats 
    \tand \Configuration_1 \in \InitialConfs, \Configuration_2,\ldots,\Configuration_n \in \Confs \\
    & \tand \Configuration_1 \OpRel \ldots \OpRel \Configuration_n
  \,\}~. \tag*{\EndDef}
\end{align*}
\end{definition}
The definition of $\Obs{}$ ignores the observable state as long as the system is not initialized, \ie, we have not reached $\InitializeBlock{\ldots}$ yet.
This is justified as long as observable state is never modified before initialization
and, thus,
the implementation performs no observable action that needs to be matched by the abstract model.
Our proof system will guarantee this property.

\subsection{Proof System}
\label{sec:methodology:logic}
Recall from \secref{sec:sound:logic} the proof
system of concurrent separation logic (CSL).
We now extend CSL to a sound program logic for proving
refinement in the extended programming language introduced earlier in this section.
Our assertion language is the same as for CSL, \ie, it consists of all $\SL$ assertions (see \secref{sec:sound:assertions}).

\paragraph{Judgments.}
Judgments that can be derived in our logic
(notice the $\vdash_{\refineIndicator}$ to distinguish our judgments from CSL judgments) 
 are of the form
\begin{align*}
  \srtriple{\PP}{\CC}{\PQ}~,
\end{align*}
where
\LockEnv is a lock environment such that $\bigstar_{\Lock \in \Locks} \LockEnv(\Lock)$
describes the shared state,
$\PP \in \SL$ is the precondition evaluated in the local state,
\CC is a command
and
$\PQ \in \SL$ is the postcondition evaluated in the local state.

\figref{fig:proof-rules-novel} shows the proof rules for the new commands $\Initialize$, $\Next$, and $\Print{\EE}$.
Furthermore, our proof system inherits all proof rules from CSL (\cf~\figref{fig:proof-rules} where we tacitly replace $\vdash$ by $\vdash_{\refineIndicator}$ and exclude dedicated
ghost locks in the rules $\Rule{Lock}$ and $\Rule{With}$).
We briefly go over the rules for the new commands:
\begin{figure}[t]
  \centering
  \input{fig-proof-rules-novel}

  \caption{Proof rules for the new commands.
    Here, $\vec{o} = (o_1,\ldots,o_\AbsLength)$ are fresh variables.
    Moreover, for some $\Perm > 0$,
    $\Aabsvars{\vec{y}} \defeq
    \apt{\AbsLoc{1}}{\Perm}{y_1} \sep \ldots \sep 
    \apt{\AbsLoc{\AbsLength}}{\Perm}{y_\AbsLength}$.
  }
  \label{fig:proof-rules-novel}
\end{figure}
\paragraph{Initialization Blocks.}
The rule \Rule{Init} can be viewed as a stronger version of the rule \Rule{Lock},
which declares \emph{two} locks at once---the ghost lock $\GhostLock$ with 
lock invariant $\GhostInv$ and a second dedicated ghost lock $\InitLock$ with an empty lock
invariant.

The ghost lock $\GhostLock$ has already been used in our operational semantics; it governs
access to the state of the abstract model \AbsTS.
The lock $\InitLock$ is never acquired but serves as a marker in our proof system
to record that an initialization block has been reached; we will explain why we need it in
more detail further below.

\Rule{Init} is a stronger version of \Rule{Lock}; in addition to declaring a lock, it verifies that the abstract model's state, \ie, the contents
$y_1,\ldots,y_\AbsLength$ of the ghost addresses $\AbsLocs$, is a valid initial state in \AbsTS, \ie,
$\AbsInit(y_1,\ldots,y_\AbsLength)$ holds.
By \assref{ass:ghostlock}, we know that the ghost lock's invariant $\GhostInv$
holds the necessary permissions to access the contents of these ghost addresses.

Regarding the second ghost lock $\InitLock$,
notice that the output stream, that is, the observable content of $\stdout$, needs to be 
part of the program state before initialization, since the ghost address $\stdout$ cannot be allocated by the implementation itself.
In \defref{def:cmd:traces}, we formalized the traces of an implementation in a way such that the content of $\stdout$ is ignored before initialization.
The underlying rationale is that the implementation performs no observable action before it is initialized and we know the state of the abstract model \AbsTS.
To guarantee that there are indeed no modifications of observable state---in our case: there are no $\Print{\EE}$ calls---before initialization, we use the ghost lock $\InitLock$.
It acts as a token that is created upon initialization,
stays in the lock environment, since it is never acquired
(in contrast to the ghost lock $\GhostLock$), and needs to be part of the lock environment \LockEnv for all rules that modify observable state, such as \Rule{Print}.
\paragraph{Next Blocks.}
The rule \Rule{Next} can be viewed as a specialized version of the rule \Rule{With}, applied to ghost lock $\GhostLock$ and wait condition \texttt{true}. However, it contains an additional premise to ensure that
executing $\NextBlock{\CC}$ can be simulated by taking a \emph{single} transition of \AbsTS. Since $\CC$ is atomic, the postcondition $\AbsNext(\vec{o},\vec{y})$ in the premise achieves this effect (where $\vec{o}$ and $\vec{y}$
capture the state of $\AbsTS$ before and after execution of $\NextBlock{\CC}$).
For initialized commands, \assref{ass:ghostlock} guarantees that changes to the transition
system's state are possible only when holding the lock $\GhostLock$, \ie, inside
of $\NextBlock{\CC}$ blocks; outside of such blocks, $\AbsTS$ will
perform a stutter step for every step of the concrete implementation.

\paragraph{Print Statements.}
Finally, to understand the \Rule{Print} rule,
we notice that $\Print{\EE}$ is---if we ignore atomicity---syntactic sugar for the command
\[
  \Command \ddefeq 
  \Lookup{\Var}{\stdout}\Seq\Mutation{\stdout}{\Append{\EE}{\Var}},
\]  where $\Var$
is a fresh variable that is not used anywhere else.
Assuming that $\Print{\EE}$ and $\Command$ are identical, 
the rule \Rule{Print} can then be derived using the standard rules in \figref{fig:proof-rules} (in sequential separation logic, \ie, for $\LockEnv = \emptyset$).
\subsection{Soundness}
\label{sec:sound:proof}
Intuitively, soundness means that deriving $\srtriple{\PP}{\CC}{\PQ}$ implies that, for every execution of $\CC$
that starts in a configuration with a local state given by $\PP$ and a shared state given by $\LockEnv$, 
the trace of the execution is also a trace of the abstract model $\AbsTS$.
We now formalize and prove the above soundness claim for our refinement logic.
We build upon Vafeiadis'~\cite{Vafeiadis11} existing soundness proof for CSL (\cf~\secref{sec:preliminaries}).
Hence, we first restate the main ingredients for proving CSL sound, where we slightly generalize to simplify their re-use.

Let $\Succ(\Configuration,\Configuration')$ be a predicate over \emph{two} program configurations; we will use $\Succ(\Configuration,\Configuration')$ only if there is a transition 
$\Configuration \OpRel \Configuration'$.
Soundness of CSL as well as our refinement logic is based
on the following notion of configuration safety:
\begin{definition}[Generalized Configuration Safety]
\label{def:conf-safe}
The $n$-step safety predicate
$\Safe{\Succ}{n}{\CC}{\Stack}{\PHeap}{\LockEnv}{\PQ}$ 
is recursively defined as follows:
  \begin{itemize}
    \item $\Safe{\Succ}{0}{\CC}{\Stack}{\PHeap}{\LockEnv}{\PQ}$ holds always.
    \item $\Safe{\Succ}{n+1}{\CC}{\Stack}{\PHeap}{\LockEnv}{\PQ}$ holds if and only if
        \begin{enumerate} 
            \item $\Reads{\CC}{\Stack} \subseteq \dom{\PHeap}$;
            \item for all $\HeapF$, if $\isDefined{\PHeap \Add \HeapF}$, then $(\Conf{\CC}{\Stack}{\PHeap \Add \Heap}) \NOpRel \Abort$;
            \item if $\CC = \Skip$, then $\Stack,\Heap \models \PQ$;
            \item for all $\CC'$, $\Stack'$, $\HeapF$, $\HeapS$, $\PHeap'$, if
            $\Stack,\HeapS \models \bigstar_{\Lock \in \Locked{\CC'}\setminus\Locked{\CC}} \LockEnv(\Lock)$ \\
            and
            $(\Conf{\CC}{\Stack}{\PHeap \Add \HeapS \Add \HeapF}) \OpRel (\Conf{\CC'}{\Stack'}{\PHeap'})$, then there exist $\PHeap''$ and $\HeapS'$ such that
           \begin{enumerate}
             \item $\PHeap' = \PHeap'' \Add \HeapS' \Add \HeapF$;
             \item $\Stack',\HeapS' \models \bigstar_{\Lock \in \Locked{\CC}\setminus\Locked{\CC'}} \LockEnv(\Lock)$;
             \item $\Succ(\CC,\Stack,\PHeap\Add\HeapS\Add\HeapF,\CC',\Stack',\PHeap')$ holds; and
             \item $\Safe{\Succ}{n}{\CC'}{\Stack'}{\PHeap''}{\LockEnv}{\PQ}$ holds.
  \EndDef
           \end{enumerate}
        \end{enumerate} 
  \end{itemize}
\end{definition}
Intuitively, $\Safe{\Succ}{n}{\CC}{\Stack}{\PHeap}{\LockEnv}{\PQ}$ means that every configuration $\Configuration$ with command $\CC$, 
local state $\PHeap$ and
shared state $\HeapS$ given by $\LockEnv$, and, optionally,
 additional global state $\HeapF$,
is \emph{safe} w.r.t. lock environment $\LockEnv$ and postcondition $\PQ$ for $n \geq 0$ steps; that is,
every configuration $\Configuration'$ reached from $\Configuration$ via one transition
(1) has no data race,
(2) does not abort execution,
(3) satisfies postcondition $\PQ$ if execution terminated,
(4 a,b) respects the lock invariants in $\LockEnv$, 
(4 c) if we reached $\Configuration'$ via a transition $\Configuration'' \OpRel \Configuration'$,
then $\Succ(\Configuration'', \Configuration')$ holds, and
(4 d) is itself safe for $n-1$ steps.

\begin{definition}[Valid CSL Triples]\label{def:valid:csl}
We write $\semtriple{}{\LockEnv}{\PP}{\CC}{\PQ}$ (read: the triple $\{\PP\}\,\CC\,\{\PQ\}$ is \emph{valid} given lock environment $\LockEnv$)
if and only if
\[
    \forall n, \Stack, \PHeap\colon \quad 
    \Stack,\PHeap \models \PP \qimplies 
    \Safe{\True}{n}{\CC}{\Stack}{\Heap}{\LockEnv}{\PQ}~.
\]
\end{definition}
Vafeiadis~\cite{Vafeiadis11} used the above notions of configuration safety
and validity to prove that every CSL judgment that can be derived using the 
rules in \figref{fig:proof-rules} yields a valid triple; the original soundness proof 
for CSL is due to Brookes~\cite{Brookes07}.
Formally:
\begin{theorem}[Soundness of CSL]\label{thm:csl-sound} 
  If $\striple{\PP}{\CC}{\PQ}$, then $\semtriple{}{\LockEnv}{\PP}{\CC}{\PQ}$.
\end{theorem}
Compared to CSL, our refinement logic introduces one rule 
for each of the new commands $\InitializeBlock{\CC}$, $\NextBlock{\CC}$, and $\Print{\EE}$ (see \figref{fig:proof-rules-novel}).
If we only consider configuration safety as in CSL, 
then, as discussed in \secref{sec:methodology:logic},
each of these rules is either a special case of the existing rules (\Rule{Print}) or a variant of an existing rule
with an even stronger precondition (\Rule{Init}, \Rule{Next}).
Hence, the soundness of CSL carries over to our logic 
 (recall that $\vdash_{\refineIndicator}$ indicates refinement judgments):
\begin{lemma}\label{lem:safety}
  If $\srtriple{\PP}{\CC}{\PQ}$, then $\semtriple{}{\LockEnv}{\PP}{\CC}{\PQ}$.
\end{lemma}
\appref{app:refine-found} provides further details.
To formalize that the implementation refines the abstract model \AbsTS, we extend the above definition of configuration safety such that, if the system is initialized, every concrete transition $\Configuration \OpRel \Configuration'$ taken is simulated by a transition of $\AbsTS$ (or stutters).
Moreover, if the system becomes initialized after a step, we require that the encoded state of $\AbsTS$ is a valid initial state.

Our formalization uses the state extraction function $\AbsLocsOfName$ from \defref{def:methodology:get-state} to refer to the abstract model's state encoded in the heap.
We then consider the configuration safety predicate
$\Safe{\RSucc}{n}{\CC}{\Stack}{\PHeap}{\LockEnv}{\PQ}$,
where the predicate $\RSucc((\CC,\Stack,\PHeap),(\CC',\Stack',\PHeap'))$ holds iff
  \begin{enumerate}[(i)]
    \item if $\isInit{\CC}$ holds, then 
    $\valid \AbsNext(\AbsLocsOf{\PHeap},\AbsLocsOf{\PHeap'})$;
    \item if $\isInit{\CC}$ does not hold and $\isInit{\CC'}$ holds, then $\valid \AbsInit(\AbsLocsOf{\PHeap'})$.
  \end{enumerate}
\begin{definition}[Valid Refinement Triples]\label{def:valid:refinement}
We write $\rsemtriple{\LockEnv}{\PP}{\CC}{\PQ}$ (read: the triple $\{\PP\}\,\CC\,\{\PQ\}$ is \emph{valid for refinement of \AbsTS} given $\LockEnv$)
if and only if
\[
    \forall n, \Stack, \PHeap\colon \quad 
    \Stack,\PHeap \models \PP \qimplies 
    \Safe{\RSucc}{n}{\CC}{\Stack}{\Heap}{\LockEnv}{\PQ}~.
\]
\end{definition}
As for CSL, our refinement logic is sound in the sense
that formally derived judgments yield valid triples: 
\begin{theorem}\label{thm:refine-sound}
  If $\srtriple{\PP}{\CC}{\PQ}$, then $\rsemtriple{\LockEnv}{\PP}{\CC}{\PQ}$.
\end{theorem}
\begin{proof}[Sketch]
The proof is by structural induction over the rules of our refinement logic
(found in \figref{fig:proof-rules} and \figref{fig:proof-rules-novel}).
In each case, we show that the predicate
$\Safe{\RSucc}{n+1}{\CC}{\Stack}{\Heap}{\LockEnv}{\PQ}$ holds;
we first invoke \lemref{lem:safety} such that only two proof obligations remain
for every possible step 
\begin{align*}
  \Configuration = (\Conf{\CC}{\Stack}{\PHeap \Add \HeapS \Add \HeapF}) \OpRel (\Conf{\CC'}{\Stack'}{\PHeap'' \Add \HeapS' \Add \HeapF}) = \Configuration'
  \tag{$\dag$}
\end{align*}
considered in \defref{def:conf-safe}.(4):
(a) $\RSucc(\Configuration,\Configuration')$ 
and
(b) $\Safe{\RSucc}{n}{\CC'}{\Stack'}{\Heap''}{\LockEnv}{\PQ}$.
There are three main cases:
First, for axioms and \Rule{Next}, 
we have $\CC' = \Skip$; then (a) is immediate, since
$\neg\isInit{\CC}$ and $\neg\isInit{\CC'}$ hold; (b) follows from 
\lemref{lem:soundness-skip} below.
Second, for the \Rule{Init} rule, (a) is immediate by the rule's premise;
(b) follows from \lemref{lem:soundness-ghost-lock} below.
Third, and finally, for all other rules, (a) and (b) follow directly from the induction
hypothesis.
Further details are found in \appref{app:refine-sound}.
\qed
\end{proof}
\begin{lemma}\label{lem:soundness-skip}
  For all $n,\Stack,\Heap,\LockEnv,\PQ$, if $\Stack,\Heap \models \PQ$, then
  $\Safe{\RSucc}{n}{\Skip}{\Stack}{\PHeap}{\LockEnv}{\PQ}$.
\end{lemma}
\begin{lemma}\label{lem:soundness-ghost-lock}
If $\CC$ contains no two different sub-commands $\CC_1$ and $\CC_2$ such that
   $\Locked{\CC_1} \cap \Locked{\CC_2} \neq \emptyset$
and $\FV{\GhostInv} \cap \Mod{\Command} = \emptyset$, then
\[ 
    \rtriple{}{\LockEnvG}{\PP}{\Command}{\PQ}
    \qimplies
    \rsemtriple{\LockEnv}{\PP \sep \GhostInv}{\DeclareLock{\GhostLock}\Block{\Command}}{\PQ \sep \GhostInv}.
\]
\end{lemma}
We will use validity in the proof of trace inclusion.
Furthermore, we need to define the initial configurations that
determine the traces of the concrete implementation. 
We consider those traces of the concrete implementation
that start in a (global) state that at least covers the local state specified
by the precondition $\PP$ and the shared state described by the lock invariants in
$\LockEnv$.
\begin{definition}[Initial Configurations]\label{def:sound-initials}
  For a command $\CC$, an assertion $\PP$, and a lock environment $\LockEnv$, 
  the set $\SoundInitials$ of \emph{initial configurations} is
  \begin{align*}
     \SoundInitials \ddefeq \{ (\Conf{\CC}{\Stack}{\PHeap \Add \HeapS \Add \HeapF}) ~|~ 
    & \isDefined{\PHeap \Add \HeapS \Add \HeapF} \tand \Stack,\PHeap \models \PP \\
    & \tand \Stack, \HeapS \models \bigstar_{\Lock \in \Locks \setminus \Locked{\CC}} \LockEnv(\Lock)  \}~.
  \end{align*}
\end{definition}
To guarantee trace inclusion, we exclude commands that suddenly become ``uninitialized'' because the $\InitializeBlock{\CC}$ block expires but execution continues, \eg, $\InitializeBlock{\CC'}\Seq\Assign{\AbsLoc{2}}{17}\Seq\Initialize\Block{\CC''}$. Such ill-formed commands can easily be detected by a syntactic check.
More precisely, we call a command $\CC$ \emph{continuously initialized} if 
there exist commands $\CC_1,\CC_2$ such that $\CC_1$ contains no sub-command of the form
$\InitializeBlock{\ldots}$ and $\CC = (\CC_1 \Seq \InitializeBlock{\CC_2})$.
\begin{theorem}[Soundness]\label{thm:trace-inclusion}
 For every continuously-initialized command $\CC$, 
 \[ \srtriple{\PP}{\CC}{\PQ} \qimplies \Traces{\SoundInitials} \subseteq \AbsTraces~. \]
\end{theorem}
A detailed proof is found in \appref{app:trace-inclusion}.

%% file: fig-operational-novel.tex
\begin{tabular}{l@{\hskip 0.05in}l@{\hskip 0.05in}l}
\Rule{Print}
&
$\Conf{\Print{\EE}}{\Stack}{\Heap}
 \OpRel
 \Conf{\Skip}{\Stack}{\Heap\Update{\OutVar}{\Append{\Sem{\EE}(\Stack)}{\Heap(\OutVar)}}}$
& if $\OutVar \in \dom{\Heap}$ 
\\
\Rule{PrintA}
&
$\Conf{\Print{\EE}}{\Stack}{\Heap}
 \OpRel \Abort$
& if $\OutVar \notin \dom{\Heap}$ 
\\
\Rule{Init}
&
$\Conf{\InitializeBlock{\Command}}{\Stack}{\Heap}
 \OpRel
 \Conf{\DeclareLockBlock{\GhostLock}{\Command}}{\Stack}{\Heap}$
& 
\\[0.5em]
\multicolumn{3}{c}{
\inference[\Rule{Next}]{
  \isAtomic{\Command} &
 \Conf{\Command}{\Stack}{\Heap}
  \OpRelTrans
  \Conf{\Skip}{\Stack'}{\Heap'}
}{
\Conf{\NextBlock{\Command}}{\Stack}{\Heap}
  \OpRel
  \Conf{\Skip}{\Stack'}{\Heap'}
}
}
\end{tabular}

%% file: fig-proof-rules-novel.tex
\begin{tabular}{l}
\inference[\Rule{Init}]{
  \rtriple{}{\LockEnv,\LockDef{\GhostLock}{\GhostInv},\LockDef{\InitLock}{\emp}}{
    \PP
  }{
    \Command
  }{
    \PQ
  }
}{
  \rtriple{}{\LockEnv}{
  \left(\exists \vec{y} \mydot \Aabsvars{\vec{y}} 
        \wedge \AbsInit(\vec{y}) \wedge \GhostInv\right) \sep \PP
  }{
      \InitializeBlock{\Command}
  }{
      \GhostInv \sep \PQ 
  }
}
\\[1.5em]
\inference[\Rule{Next}]{
  \isAtomic{\CC} \\
  \srtriple{
      (\Aabsvars{\vec{o}} \wedge \GhostInv) \sep \PP
  }{
      \Command
  }{
      (\exists \vec{y} \mydot \Aabsvars{\vec{y}} \wedge \AbsNext(\vec{o},\vec{y}) \wedge
      \GhostInv) \sep \PQ)
  }
}{
  \rtriple{}{\LockEnv,\LockDef{\GhostLock}{\GhostInv}}{
      \PP
  }{
      \NextBlock{\Command}
  }{
      \PQ
  }
}
\\[1.25em]
\inference[\Rule{Print}]{
}{
  \rtriple{}{\LockEnv,\LockDef{\InitLock}{\emp}}{
        \pt{\OutVar}{\PermWrite}{\EE'}
    }{
        \Print{\EE}
    }{
        \pt{\OutVar}{\PermWrite}{(\Append{\EE}{\EE'})}
    }
}
\end{tabular}

%% file: evaluation.tex
To evaluate our verification technique, we verified seven case studies 
using the separation logic-based automated verifier
Viper~\cite{MuellerSchwerhoffSummers16}. All of our examples verify using Viper's existing verification backends, demonstrating that our methodology is supported by readily available verification tools. 
We will make our case studies available as an artifact.
The examples show the flexibility of our proofs in the four dimensions mentioned in~\secref{sec:introduction} and are summarized in~\tabref{tab:evaluation}. We briefly describe each of the examples in the following.

\input{tab-evaluation}

\texttt{alternating} demonstrates how multiple threads can collaborate to achieve the overall behavior of an abstract system. This example is described in~\secref{sec:overview}.

\texttt{barrier} is an adaptation of the barrier example from Armada~\cite{LorchCKPQSWZ20}, originally from Cohen and Schirmer~\cite{CohenSchirmer2010}. In this example, multiple threads are started and must pass a barrier before exiting. A guarded transition system is used to ensure that each thread can only perform transitions related to its own state.

\texttt{cons\_producer} is an example with a consumer-producer setup, in which one thread adds values to a shared buffer and the other consumes them. We show that at all points an invariant is maintained, namely the relationship between the number of values produced, the number of values consumed, and the number of values left in the shared buffer.
\texttt{echo\_server} demonstrates the use of I/O methods to model both standard input and standard output.

\texttt{ring\_leader} implements leader election in a ring, with a \TLA specification adapted from the \TLA examples repository, based on Chang and Roberts~\cite{ChangRoberts1979}.

\texttt{trees\_product} and \texttt{trees\_record} are both examples demonstrating data refinement: the abstract model uses a mathematical datatype to represent a tree of values, whereas the implementation uses a heap-allocated array representation of the tree. The \texttt{trees\_record} example further demonstrates dynamic threading to process the input tree in parallel, and a two-state monotonicity invariant.

\paragraph{Discussion.}
Our case studies demonstrate that our methodology is highly flexible w.r.t.\ to the structure of both programs and proofs. In particular, they use local and shared mutable state (including dynamic heap data structures), concurrency with dynamic thread creation and synchronization via locks and barriers, as well as three different reasoning idioms.

The evaluation also shows that our methodology enables automating refinement proofs using an off-the-shelf verification tool. While we used Viper as a concrete tool, no example relies on features that are genuinely Viper-specific, which supports the claim that our approach is flexible w.r.t.\ to the underlying logic. The verification time for each example is below 8s, which demonstrates that our methodology is well-suited for SMT-based automation.

%% file: tab-evaluation.tex
\begin{table}[t]
\caption{Case studies used for evaluation. \textbf{Data ref.} indicates whether there is interesting data refinement between the abstract model and the implementation. \textbf{Threads} indicates the number of nodes or worker threads in the implementation, where a $*$ means the threads are spawned dynamically after the model was already initialized. \textbf{Sync.} indicates the kind of synchronization primitive used, if any. \textbf{Idiom} indicates the reasoning used on top of our methodology. \textbf{SLOC} indicates standard lines of code including annotations. \textbf{Time} indicates verification time in seconds, measured as an average of the wall-clock runtime over 10 runs using Viper's symbolic execution verification backend on an Intel Core i9-10885H 2.40GHz CPU with 16 GiB of RAM.}
\label{tab:evaluation}
\centering
\begin{adjustbox}{center}
\begin{tabular}{l@{\hskip 0.15in}c@{\hskip 0.15in}c@{\hskip 0.15in}c@{\hskip 0.15in}l@{\hskip 0.15in}l@{\hskip 0.15in}r@{\hskip 0.15in}r}
  \hline
  \hline
  \textbf{Example}
    & \textbf{Data ref.}
    & \textbf{Threads}
    & \textbf{Sync.}
    & \textbf{Idiom} 
    & \textbf{SLOC}
    & \textbf{Time}
  \\
  \hline
  \hline
  \texttt{alternating}
    & ---
    & 2
    & Lock
    & Owicki-Gries
    & 118
    & 3.78 \\
  \texttt{barrier}
    & ---
    & $N$
    & Barrier
    & Guards
    & 383
    & 6.88 \\
  \texttt{cons\_producer}
    & ---
    & 2
    & Lock
    & Guards
    & 213
    & 4.15 \\
  \texttt{echo\_server}
    & ---
    & 1
    & ---
    & ---
    & 69
    & 3.39 \\
  \texttt{ring\_leader}
    & ---
    & $N$
    & ---
    & ---
    & 279
    & 7.74 \\
  \texttt{trees\_product}
    & \checkmark
    & $N*$
    & ---
    & ---
    & 192
    & 3.48 \\
  \texttt{trees\_record}
    & \checkmark
    & $N*$
    & ---
    & Rely-guarantee
    & 268
    & 4.45 \\
  \hline
  \hline
\end{tabular}
\end{adjustbox}
\end{table}

%% file: related_work.tex
In this section, we survey refinement techniques that combine abstract models and executable code. 

Various approaches~\cite{LesaniBC16,RahliVVV18,SergeyWT18,WoosWATEA16} develop implementations that are correct by construction by refining abstract models within Coq and then extracting executable OCaml programs. Similarly, Liu et al.~\cite{nfm20-liu} model distributed systems in Maude's rewriting logic and compile them into  implementations running in distributed Maude sessions. The code extracted by these approaches is typically sub-optimal (for instance, does not use mutable data structures) and cannot interface with existing libraries, which is often necessary in practice. In contrast, our methodology uses bottom-up verification and can handle efficient implementations using concurrency and mutable state.

Trillium~\cite{Trillium} is a refinement technique based on separation logic. Like our methodology, it supports a wide range of program and data structures. Trillium is based on Iris~\cite{Jung2018} and formalized in Coq, which enables foundational proofs at the expense of substantial manual effort. Trillium inherits some of Iris's limitations. In particular, it is limited to finitary behaviors and, thus, does not support the common case that abstract models choose a value non-deterministically from an infinite set. Moreover, Trillium expresses coupling invariants via Iris's invariants, which complicates reasoning about system initialization, especially allocation. Our methodology does not have these limitations. Trillium supports liveness properties, which we do not handle yet.

Armada~\cite{LorchCKPQSWZ20} supports the verification of concurrent, high-performance code written in a C-like language. To achieve refinement against an abstract model, the user specifies a sequence of steps to gradually transform the implementation into the specification. Non-trivial refinement steps require complex Dafny~\cite{Leino10} proofs showing a connection between two state machines. Unlike Armada, our methodology does not convert programs to state machines and the coupling between the abstract model and the implementation can be much looser.

The CIVL verifier~\cite{Hawblitzel2015,KraglQH20} also organizes the refinement proof into multiple layers. Each layer is a \emph{structured concurrent program}, where the concurrent behavior is reflected in the program structure. This structure simplifies the proof obligations and allows automation, but also reduces program flexibility. Refinement steps are based on a set of  trusted tactics. By contrast, our methodology imposes no restrictions on the program or proof structure.

Igloo~\cite{Igloo} connects abstract models to concrete implementations via dedicated I/O specifications~\cite{Penninckx0P15}. Similarly to our work, they support a variety of separation logics to reason about concrete implementations. However, their technique does not allow threads to perform I/O operations concurrently, whereas our methodology has no such limitation.

Similar to our methodology, IronFleet~\cite{HawblitzelHKLPR15} embeds abstract models as ghost state into executable programs and automates verification using an SMT-based verifier, in their case Dafny. However, their refinement technique imposes severe restrictions on the executable program. It must be sequential and its structure must mirror the structure of the abstract model. IronFleet supports both safety and liveness properties,
whereas our approach focuses on safety properties and leaves liveness as future work.

The refinement technique~\cite{Koh0LXBHMPZ19} used in  DeepSpec~\cite{DeepSpecPosition} is based on the Verified Software Toolchain~(VST)~\cite{CaoBGDA18}, a framework for verifying C programs via a separation logic embedded in Coq. Instead of transition systems, they specify the intended system behavior using interaction trees~\cite{XiaZHHMPZ20}, which are embedded into VST's separation logic. In contrast, our methodology allows us to apply standard separation logics and existing program verifiers.

Oortwijn and Huisman~\cite{OortwijnH19} embed process calculus models into a concurrent SL, which is automated using Viper. Their refinement approach preserves state assertions, but it is unclear whether arbitrary trace properties are preserved.

%% file: conclusion.tex
We have introduced a methodology for refinement proofs in separation logic that is flexible in terms of the type of abstract model used, the structure of the concrete implementation, the underlying logic and tool chain, as well as the structure of the proofs themselves. We have formalized the methodology on top of concurrent separation logic
and demonstrated its applicability on several case studies using the automated verifier Viper.
As future work, we plan to extend our methodology to liveness properties.

%% file: appendix.tex
\section{Quick Reference}

As a quick reference, \tabref{tab:notation} summarizes the notational conventions that have
been introduced in \secref{sec:preliminaries} (above the line) and \secref{sec:methodology} (below the line).
\input{tab-notation}


\section{Lemmas from CSL}
We will use the following lemmas taken from~\cite{Vafeiadis11}:
\begin{lemma}\label{lem:aux:skip}
  For all $n,\Stack,\Heap,\LockEnv,\PQ$, if $\Stack,\Heap \models \PQ$, then
  $\Safe{\True}{n}{\Skip}{\Stack}{\PHeap}{\LockEnv}{\PQ}$.
\end{lemma}

\begin{lemma}\label{lem:aux:atom}
  If $\forall n\colon \Safe{\True}{n}{\CC}{\Stack}{\PHeap}{\LockEnv}{\PQ}$
  and $\isDefined{\PHeap \Add \HeapF}$ and
  \[ 
    \Conf{\CC}{\Stack}{\PHeap \Add \HeapF} \OpRelTrans \Conf{\Skip}{\Stack'}{\Heap'},
  \]
  then there exists $\PHeap''$ such that $\Heap' = \PHeap'' \Add \HeapF$ and 
  $\Stack', \PHeap'' \models \PQ$.
\end{lemma}

\section{Proof of \lemref{lem:safety}}\label{app:refine-found}
\subsubsection{Claim.}
  If $\srtriple{\PP}{\CC}{\PQ}$, then $\semtriple{}{\LockEnv}{\PP}{\CC}{\PQ}$.

\begin{proof}
  By \defref{def:valid:csl}, it suffices to show that
\begin{align}
     & \srtriple{\PP}{\CC}{\PQ} \qimplies \\
     & \qquad \forall n, \Stack, \PHeap\colon 
       \qquad  \Stack,\PHeap \models \PP \qimplies
       \Safe{\True}{n}{\CC}{\Stack}{\Heap}{\LockEnv}{\PQ}~. \notag
\end{align}
Since $\Safe{\True}{0}{\CC}{\Stack}{\PHeap}{\LockEnv}{\PQ}$ holds always by \defref{def:conf-safe}, it suffices to prove that
$\Safe{\True}{n+1}{\CC}{\Stack}{\PHeap}{\LockEnv}{\PQ}$ holds
for an arbitrary, but fixed $n \in \Nats$.

  By structural induction over the rules for of our refinement logic for deriving judgments $\srtriple{\PP}{\CC}{\PQ}$
  (found in \figref{fig:proof-rules} and \figref{fig:proof-rules-novel}).

We present the case $\Rule{Print}$ in detail further below.
As discussed in \secref{sec:methodology:logic},
the rule \Rule{Init} is a variant of the CSL rule \Rule{Lock} with a stronger precondition, and the rule \Rule{Next} is a special case of the CSL rule \Rule{With}; we thus omit detailed proofs.
since all other rules are identical to the CSL proof rules; their proof is thus completely analogous to the proof of \thmref{thm:csl-sound}.

\paragraph{The case \Rule{Print}.}
Recall the rule
\[
\inference[\Rule{Print}]{
}{
  \rtriple{}{\LockEnv,\LockDef{\InitLock}{\emp}}{
        \pt{\OutVar}{\PermWrite}{\EE'}
    }{
        \Print{\EE}
    }{
        \pt{\OutVar}{\PermWrite}{(\Append{\EE}{\EE'})}
    }
}
\]
and assume
\begin{align}
\Stack,\PHeap \models \pt{\OutVar}{\PermWrite}{\EE'}.
\end{align}
Consequently, 
\begin{align}
\PHeap \eeq \Set{ \stdout \xmapsto{\PermWrite} \EE'(\Stack) }.
\label{eq:cs:3}
\end{align}
We discharge all items in \defref{def:conf-safe} to show that
\begin{align}
  \Safe{\True}{n+1}{\Print{\EE}}{\Stack}{\PHeap}{\LockEnv,\LockDef{\InitLock}{\emp}}{\pt{\OutVar}{\PermWrite}{(\Append{\EE}{\EE'})}}.
\end{align}
For \saferef{1}, consider the following:
\begin{align}
  & \Reads{\Print{\EE}}{\Stack} = \Set{\stdout} \\
 \eeq & \dom{\Heap} \tag{by \eqref{cs:3}}
\end{align}
\saferef{2} is immediate by \eqref{cs:3} and the rules of our operational semantics.
\saferef{3} is trivial.
For \saferef{4}, assume $\CC'$, $\Stack'$, $\HeapF$, $\HeapS$, $\PHeap'$ such that
\begin{align}
  & \Stack,\HeapS \models \bigstar_{\Lock \in \Locked{\CC'}\setminus\Locked{\CC}} \LockEnv(\Lock)
\\
   \tand \notag \\
  &
  (\Conf{\Print{\EE}}{\Stack}{\PHeap \Add \HeapS \Add \HeapF}) \OpRel (\Conf{\CC'}{\Stack'}{\PHeap'})
\end{align}
By our operational semantics and \eqref{cs:3}, there is only one transition as above.
For this transition, we have
\begin{align}
  & \CC' = \Skip, \qand
  \HeapS = \emp, \qand \\
  & \Heap' = \Set{ \stdout \xmapsto{\PermWrite} \Append{\EE(\Stack)}{\EE'(\Stack)} } \Add \HeapF
\end{align}
Hence, for $\PHeap'' = \Heap'$ and $\HeapS' = \emp$, 
\saferef{(a)-(c)} hold immediately.
\saferef{(d)} holds by \lemref{lem:aux:skip}.
\qed
\end{proof}

\section{Missing lemmas for the proof of \thmref{thm:refine-sound}}\label{app:refine-sound}
\subsection{Proof of \lemref{lem:soundness-skip}}
\subsubsection{Claim.}
For all $n,\Stack,\Heap,\LockEnv,\PQ$, if $\Stack,\Heap \models \PQ$, then
  $\Safe{\RSucc}{n}{\Skip}{\Stack}{\PHeap}{\LockEnv}{\PQ}$.
\begin{proof}
  By definition, $\Safe{\RSucc}{0}{\Skip}{\Stack}{\PHeap}{\LockEnv}{\PQ}$ holds always.
  Assume 
  \begin{align}
    \Stack,\Heap \models \PQ.
  \end{align}
  By \lemref{lem:aux:skip}, we know that
  \begin{align}
    \Safe{\True}{n+1}{\Skip}{\Stack}{\PHeap}{\LockEnv}{\PQ}.
    \label{eq:sk}
  \end{align}
  Then $\Safe{\RSucc}{n+1}{\Skip}{\Stack}{\PHeap}{\LockEnv}{\PQ}$ holds as well, because
  \saferef{(1-3)} holds by \eqref{sk}.
  Furthermore, since our operational semantics does not admit any transition starting with
  command \Skip, \saferef{4} holds vacuously.
  \qed
\end{proof}

\subsection{Proof of \lemref{lem:soundness-ghost-lock}}
\subsubsection{Claim.}
If $\CC$ contains no two different sub-commands $\CC_1$ and $\CC_2$ such that
   $\Locked{\CC_1} \cap \Locked{\CC_2} \neq \emptyset$
and $\FV{\GhostInv} \cap \Mod{\Command} = \emptyset$,
\[ 
    \rtriple{}{\LockEnvG}{\PP}{\Command}{\PQ}
    \qimplies
    \rsemtriple{\LockEnv}{\PP \sep \GhostInv}{\DeclareLock{\GhostLock}\Block{\Command}}{\PQ \sep \GhostInv}.
\]
\begin{proof}[Sketch]
By structural induction on the rules of our refinement logic we show that
    $\rtriple{}{\LockEnvG}{\PP}{\Command}{\PQ}$
    and $\Stack,\PHeap \sep \GhostInv \models \PP$ implies
\begin{align}
    \forall n\in\Nats\colon \quad
    \Safe{\RSucc}{n}{\DeclareLockBlock{\GhostInv}{\CC}}{\Stack}{\Heap}{\LockEnv}{\PQ \sep \GhostInv}~.
\end{align}
We present the case for the \Rule{Next} in detail as it is the only rule, where the
predicate $\RSucc$ is not immediately discharged by either applying the induction
hypothesis or performing a stutter step.

The remaining cases are very similar to Vafeiadis~\cite{Vafeiadis11} soundness proof 
for CSL (and in particular require an additional complete induction on $n$ for some
cases, such as parallel composition).

\paragraph{The case \textnormal{\Rule{Next}}.}
Recall from \figref{fig:proof-rules-novel} the rule
\begin{align*}
\inference[\Rule{Next}]{
  \isAtomic{\CC} \\
  \srtriple{
      (\Aabsvars{\vec{o}} \wedge \GhostInv) \sep \PP
  }{
      \Command
  }{
      \underbrace{\exists \vec{y} \mydot \Aabsvars{\vec{y}} \wedge \AbsNext(\vec{o},\vec{y}) \wedge
      \GhostInv) \sep \PQ)}_{\PR}
  }
}{
  \rtriple{}{\LockEnv,\LockDef{\GhostLock}{\GhostInv}}{
      \PP
  }{
      \NextBlock{\Command}
  }{
      \PQ
  }
}
\end{align*}
and assume that
\begin{align} 
  \Stack, \PHeap \models \PP \sep \GhostInv.
  \label{eq:next1}
\end{align}
We then need to show that, for all $n \in \Nats$, 
\begin{align} 
\Safe{\RSucc}{n}{\DeclareLockBlock{\GhostLock}{\NextBlock{\CC}}}{\Stack}{\Heap}{\LockEnv}{\PQ \sep \GhostInv}.
\end{align}
The case $n = 0$ always holds by \defref{def:conf-safe}.
For $n > 0$, we discharge \saferef{(1)-(4)}:
\begin{itemize}
  \item \saferef{(1)} is immediate, since  
        $\Reads{\DeclareLockBlock{\GhostLock}{\NextBlock{\CC}}}{\Stack} = \emptyset$;
  \item \saferef{(2)} is immediate, by \lemref{lem:safety} and the rules of our operational
        semantics;
  \item \saferef{(3)} is trivial; and
  \item for \saferef{(4)} assume
            $\CC'$, $\Stack'$, $\HeapF$, $\HeapS$, $\PHeap'$ such that
        \begin{align}
      &      \Stack,\HeapS \models \bigstar_{\Lock \in \Locked{\CC'}\setminus\Locked{\CC}} \LockEnv(\Lock) \\
& \tand \notag\\
             & (\Conf{\CC}{\Stack}{\PHeap \Add \HeapS \Add \HeapF}) \OpRel (\Conf{\CC'}{\Stack'}{\PHeap'}).
        \end{align}
        Our operational semantics has exactly one rule that admits such a transition, namely
        the one with premise
        $\Conf{\Command}{\Stack}{\Heap \Add \HeapS \Add \HeapF} \OpRelTrans \Conf{\Skip}{\Stack'}{\Heap'}$.
        Now, by the \Rule{Next} rule's premise and \thmref{lem:safety}, we have
        \begin{align}
          \forall m\colon \Safe{\True}{m}{\CC}{\Stack}{\PHeap}{\LockEnv}{\PR}.
        \end{align}
        Hence, by \lemref{lem:aux:atom}, there exist $\PHeap''$ 
        such that
        \begin{align}
            \Heap' \eeq \PHeap'' \Add (\HeapS \Add \HeapF) \tand \\
            \Stack', \PHeap'' \models \PR.
        \end{align}
        Then \saferef{(4a)} holds for $\PHeap''$ as above and $\HeapS' = \EmptyHeap$;
        \saferef{(4b)} holds, since $\Locked{\DeclareLockBlock{\GhostLock}{\NextBlock{\CC}}\DeclareLockBlock{\GhostLock}{\NextBlock{\CC}}} = \Locked{\DeclareLockBlock{\GhostLock}{\Skip}}$;
        \saferef{(4c)} follows from 
  \eqref{next1} and $\Stack', (\Heap'' \Add (\HeapS \Add \HeapF)) \models \PR$, where $\PR$ is the \Rule{Next} rule's postcondition; and
        \saferef{(dc)} holds by \lemref{lem:soundness-skip}.
\end{itemize}
\qed
\end{proof}

%

\section{Proof of \thmref{thm:trace-inclusion}}\label{app:trace-inclusion}
\subsubsection{Claim.}
 For every continuously initialized command $\CC$, 
 \[ \srtriple{\PP}{\CC}{\PQ} \qimplies \Traces{\SoundInitials} \subseteq \AbsTraces~. \]

\begin{proof}
  By \thmref{thm:refine-sound}, we have
  \begin{align}
    \rsemtriple{\LockEnv}{\PP}{\CC}{\PQ}
    \label{eq:ti1}
  \end{align}
  Moreover, since $\isInit{\CC}$ holds iff $\GhostLock \in \Locks{\CC}$ and there is no proof rule for (sub-)commands of the form 
  $\DeclareLock{\GhostLock}{\ldots}$, we have
  \begin{align}
    \neg\isInit{\CC}
    \label{eq:ti2}
  \end{align}
  Now, assume that
  \begin{align}
     \Configuration_1 = (\Conf{\CC_1}{\Stack_1}{\Heap_1}) \in \SoundInitialsFor{\CC_1}
    \label{eq:ti3} 
  \end{align}
  and, for some arbitrary, but fixed, $n \in \Nats$, 
  \begin{align}
      \Configuration_1 \OpRel \Configuration_2 
      \OpRel \ldots \OpRel \Configuration_n~.
    \label{eq:ti4} 
  \end{align}
By \eqref{ti1}, \eqref{ti3}, \defref{def:sound-initials}, and \defref{def:valid:refinement} we have
\begin{align}
    \Safe{\RSucc}{n+2}{\CC}{\Stack}{\Heap}{\LockEnv}{\PQ}~.
\end{align}
By \defref{def:conf-safe}.(2) and \eqref{ti3}, we can safely assume that
\begin{align}
    \forall i \in \Set{1,\ldots,n}\colon \quad
    \Configuration_i = (\Conf{\CC_i}{\Stack_i}{\Heap_i})~.
\end{align}
We distinguish two cases:
First, assume there exists an $m' \in \Set{1,\ldots,n}$ such that $\isInit{\Command_{m'}}$;
let $m$ be the minimal such index $m'$. 
By \eqref{ti2}, we know that $m \geq 2$.
Hence, by \lemref{thm:incl-m} below, 
\begin{align}
  \AbsLocsOf{\Heap_m} \ldots \AbsLocsOf{\Heap_n} \in \AbsPaths
  \label{eq:incl}
\end{align}
Now, if $\Command_n = \Skip$, then we have:
\begin{align*}
      & \Traces{\SoundInitialsFor{\CC_1}}  \\
  ~\ni~ & \Obs{\Configuration_1} \ldots \Obs{\Configuration_n} \\
  \eeq & \underbrace{\Obs{\Configuration_1} \ldots \Obs{\Configuration_{m-1}}}_{\eeq \Emptyseq} \Obs{\Configuration_{m}} \ldots \Obs{\Configuration_{n-1}} 
 \underbrace{\Obs{\Configuration_n}}_{\eeq \Emptyseq} 
         \tag{Def. of $\Obs{.}$, $\neg\isInit{\Skip}$} \\
  \eeq & \Obs{\Configuration_{m}} \ldots \Obs{\Configuration_{n-1}} 
         \tag{Def. of $\Obs{.}$; $m$ is minimal by assumption} \\
  \eeq &
         \AbsLocsOf{\Heap_m} \ldots \AbsLocsOf{\Heap_{n-1}} \tag{Def. of $\Obs{.}$} \\
  ~\in~ & \AbsPaths. \tag{by (\ref{eq:incl})}
\end{align*}
Conversely, if $\Command_n \neq \Skip$, then we have:
\begin{align*}
      & \Traces{\SoundInitialsFor{\CC_1}}  \\
  ~\ni~ & \Obs{\Configuration_1} \ldots \Obs{\Configuration_n} \\
  \eeq & \underbrace{\Obs{\Configuration_1} \ldots \Obs{\Configuration_{m-1}}}_{\eeq \Emptyseq} \Obs{\Configuration_{m}} \ldots \Obs{\Configuration_{n-1}} \Obs{\Configuration_n} \\
  \eeq & \Obs{\Configuration_{m}} \ldots \Obs{\Configuration_{n-1}} \Obs{\Configuration_n} 
         \tag{Def. of $\Obs{.}$; $m$ is minimal by assumption} \\
  \eeq &
         \AbsLocsOf{\Heap_m} \ldots \AbsLocsOf{\Heap_n} \tag{Def. of $\Obs{.}$} \\
  ~\in~ & \AbsPaths. \tag{by (\ref{eq:incl})}
\end{align*}
Second, assume there for all $m' \in \Set{1,\ldots,n}$, we have $\neg\isInit{\Command_{m'}}$.
Then, by Def. of $\Obs{.}$, we have
\begin{align}
  \Traces{\SoundInitialsFor{\CC_1}}  
  ~\ni~ \Obs{\Configuration_1} \ldots \Obs{\Configuration_n}
  \eeq \Emptyseq
  ~\in~ \AbsPaths.
\end{align}
\end{proof}

\begin{lemma}\label{thm:incl-m}
  For all $2 \leq m \leq n$ and all 
  $(\Conf{\CC_1}{\Stack_1}{\Heap_1}), \ldots, (\Conf{\CC_n}{\Stack_n}{\Heap_n})$, if
  \begin{enumerate}[(i)]
    \item\label{item:im1} $\CC_1$ is continuously initialized and
    \item\label{item:im2} $\rsemtriple{}{\PP}{\CC_1}{\PQ}$ and
    \item\label{item:im3} $(\Conf{\CC_1}{\Stack_1}{\Heap_1}) \in \SoundInitials$ and
    \item\label{item:im4} $\Conf{\CC_1}{\Stack_1}{\Heap_1} \OpRel \ldots 
           \OpRel \Conf{\CC_n}{\Stack_n}{\Heap_n}$ and
    \item\label{item:im5} $\isInit{\CC_m}$ and, for all $j < m$, not $\isInit{\CC_j}$,
  \end{enumerate}
  then $\AbsLocsOf{\Heap_m} \ldots \AbsLocsOf{\Heap_n} \in \AbsPaths$.
\end{lemma}
\begin{proof}
By induction on $n$.
\paragraph{Induction base.} Assume \itref{im1}-\itref{im5} hold for $n = m = 2$.
By \itref{im3}, there exist heaps $\PHeap$, $\HeapS$ such that 
\begin{align}
  & \Heap_1 \eeq \PHeap \Add \HeapS \\
  & \Stack_1, \PHeap \eeq \PP \\
  & \Stack_1, \HeapS \models 
    \bigstar_{\Lock \in \Locks \setminus \Locked{\CC_1}} \LockEnv(\Lock) \\
\end{align}
Hence, by \itref{im2}, 
\begin{align}
  \Safe{\RSucc}{3}{\CC_1}{\Stack_1}{\PHeap}{\LockEnv}{\PQ}~.
\end{align}
Now, consider the transition
\begin{align}
  \Conf{\CC_1}{\Stack_1}{\PHeap \Add \HeapS \Add \EmptyHeap} \OpRel \ldots \OpRel \Conf{\CC_2}{\Stack_2}{\Heap_2}
\end{align}
which exists by \itref{im4}. By \itref{im5}, we have
$\isInit{\CC_1}$ does not hold and $\isInit{\CC_2}$ holds.
Hence, by \saferef{(4c)} and definition of $\RSucc$, 
\begin{align}
\valid \AbsInit(\AbsLocsOf{\Heap_2})
\end{align}
and thus also $\AbsLocsOf{\Heap_2} \in \AbsPaths$.

\paragraph{Induction hypothesis.}
For an arbitrary, but fixed, $n \geq 2$, assume that \itref{im1}-\itref{im5} imply
$\AbsLocsOf{\Heap_m} \ldots \AbsLocsOf{\Heap_n} \in \AbsPaths$.

\paragraph{Induction step.}
Assume that \itref{im1}-\itref{im5} hold for some $2 \leq m \leq n+1$, where $n \geq 2$.
The case $m = n + 1$ is analogous to the induction base. 
Hence, assume $m < n +1$.
By \itref{im3}, there exist heaps $\PHeap$, $\HeapS$ such that 
\begin{align}
  & \Heap_1 \eeq \PHeap \Add \HeapS \\
  & \Stack_1, \PHeap \eeq \PP \\
  & \Stack_1, \HeapS \models 
    \bigstar_{\Lock \in \Locks \setminus \Locked{\CC_1}} \LockEnv(\Lock) \\
\end{align}
Hence, by \itref{im2}, 
\begin{align}
  \Safe{\RSucc}{n+3}{\CC_1}{\Stack_1}{\PHeap}{\LockEnv}{\PQ}~.
\end{align}
By repeated unfolding the above predicate, 
we know that there exist $\PHeap'$ and $\HeapS'$ such that
\begin{align}
  & \Heap_n \eeq \PHeap' \Add \HeapS' \\
  & \Safe{\RSucc}{2}{\CC_n}{\Stack_n}{\PHeap'}{\LockEnv}{\PQ} \\
  & \Stack_1, \HeapS \models 
    \bigstar_{\Lock \in \Locks \setminus \Locked{\CC_n}} \LockEnv(\Lock) \\
\end{align}
Now, consider the transition below, which exists by \itref{im4}.
\begin{align}
  \Conf{\CC_n}{\Stack_n}{\PHeap' \Add \HeapS' \Add \EmptyHeap} \OpRel \ldots \OpRel \Conf{\CC_{n+1}}{\Stack_{n+1}}{\Heap_{n+1}}
\end{align}
Since $m < n + 1$, \itref{im1} and \itref{im5} yield that $\isInit{\CC_n}$ holds.
By  \saferef{(4c)} and definition of $\RSucc$, this means
\begin{align}
    \valid \AbsNext(\AbsLocsOf{\PHeap_n},\AbsLocsOf{\PHeap_{n+1}})
\end{align}
By I.H., we also know that
\begin{align}
  \AbsLocsOf{\Heap_m} \ldots \AbsLocsOf{\Heap_n} \in \AbsPaths
\end{align}
Hence,
\begin{align}
  \AbsLocsOf{\Heap_m} \ldots \AbsLocsOf{\Heap_n}\AbsLocsOf{\Heap_{n+1}} \in \AbsPaths,
\end{align}
which finishes the proof. \qed
\end{proof}

%% file: tab-notation.tex
{\renewcommand{\arraystretch}{1.2}
\begin{table}[t]
\caption{Notational conventions and metavariables used throughout the paper.}
\label{tab:notation}
\centering
\begin{tabular}{l@{\hskip 0.15in}l@{\hskip 0.15in}l@{\hskip 0.15in}l}
  \hline
  \hline
  \textbf{Entities} 
  & 
  \textbf{Metavariables} 
  &
  \textbf{Domain}
  &
  \textbf{Defined}
  \\
  \hline
  \hline
  Commands & \CC & \Commands & \secref{sec:sound:language} \\
  Variables & \Var, $y$, $z$, \ldots & \Vars & \secref{sec:sound:language} \\
  Expressions & \EE & $\Stacks \to \Vals$ & \secref{sec:sound:semantics} \\
  Locks & \Lock, \CLock & \Locks & \secref{sec:sound:language} \\
  Values & \Val & \Vals & \secref{sec:sound:states} \\
  Stacks & \Stack & $\Stacks$  & \secref{sec:sound:states} \\
  Heaps  & \PHeap & $\PHeaps$ & \secref{sec:sound:assertions} \\
  Normal heaps & \Heap & $\Heaps$ & \secref{sec:sound:states} \\
  Addresses & \Loc & $\Locs$ & \secref{sec:sound:states} \\
  Configurations & \Configuration & \Confs & \secref{sec:sound:semantics} \\
  Assertions & $\PP, \PQ, \PR, \ldots$ & \SL & \secref{sec:sound:assertions} \\
  Lock environments & $\LockEnv$ & & \secref{sec:sound:logic} \\
  \hline
  Ghost addresses & $\stdout$, $\AbsLocs$ & $\ALocs$ & \secref{sec:methodology:state} \\
  Ghost lock with invariant & $\LockDef{\GhostLock}{\GhostInv}$ &&  \secref{sec:methodology:locks} \\
  \hline
  \hline
\end{tabular}
\end{table}
}

%% file: ms.bbl
\begin{thebibliography}{10}
\providecommand{\url}[1]{\texttt{#1}}
\providecommand{\urlprefix}{URL }
\providecommand{\doi}[1]{https://doi.org/#1}

\bibitem{Abrial10}
Abrial, J.: Modeling in Event-B - System and Software Engineering. Cambridge
  University Press (2010)

\bibitem{DeepSpecPosition}
Appel, A.W., Beringer, L., Chlipala, A., Pierce, B.C., Shao, Z., Weirich, S.,
  Zdancewic, S.: Position paper: the science of deep specification.
  Philosophical Transactions of the Royal Society A  \textbf{375} (Oct 2017)

\bibitem{Boyland03}
Boyland, J.: Checking interference with fractional permissions. In: Cousot, R.
  (ed.) Static Analysis (SAS). pp. 55--72 (2003)

\bibitem{Brookes07}
Brookes, S.: A semantics for concurrent separation logic. Theor. Comput. Sci.
  \textbf{375}(1-3),  227--270 (2007)

\bibitem{CaoBGDA18}
Cao, Q., Beringer, L., Gruetter, S., Dodds, J., Appel, A.W.: {VST-Floyd}: {A}
  separation logic tool to verify correctness of {C} programs. J. Autom.
  Reasoning  \textbf{61}(1-4),  367--422 (2018)

\bibitem{ChangRoberts1979}
Chang, E., Roberts, R.: An improved algorithm for decentralized extrema-finding
  in circular configurations of processes. Communications of the ACM
  \textbf{22}(5),  281--283 (1979)

\bibitem{CohenSchirmer2010}
Cohen, E., Schirmer, B.: From total store order to sequential consistency: A
  practical reduction theorem. In: Kaufmann, M., Paulson, L.C. (eds.)
  Interactive Theorem Proving. pp. 403--418. Springer Berlin Heidelberg,
  Berlin, Heidelberg (2010)

\bibitem{coq}
{Coq} {Development}~{Team}, T.: The {Coq} Reference Manual, version 8.10
  (2019), available electronically at \url{http://coq.inria.fr/documentation}

\bibitem{Dinsdale-YoungDGPV10}
Dinsdale{-}Young, T., Dodds, M., Gardner, P., Parkinson, M.J., Vafeiadis, V.:
  Concurrent abstract predicates. In: D'Hondt, T. (ed.) European Conference on
  Object-Oriented Programming (ECOOP). Lecture Notes in Computer Science,
  vol.~6183, pp. 504--528. Springer (2010)

\bibitem{HawblitzelHKLPR15}
Hawblitzel, C., Howell, J., Kapritsos, M., Lorch, J.R., Parno, B., Roberts,
  M.L., Setty, S.T.V., Zill, B.: Iron{F}leet: proving practical distributed
  systems correct. In: Miller, E.L., Hand, S. (eds.) Symposium on Operating
  Systems Principles (SOSP). pp. 1--17. {ACM} (2015)

\bibitem{Hawblitzel2015}
Hawblitzel, C., Petrank, E., Qadeer, S., Tasiran, S.: Automated and modular
  refinement reasoning for concurrent programs. In: International Conference on
  Computer Aided Verification. pp. 449--465. Springer (2015)

\bibitem{Jones81}
Jones, C.B.: Developing methods for computer programs including a notion of
  interference. Ph.D. thesis, University of Oxford, {UK} (1981),
  \url{http://ethos.bl.uk/OrderDetails.do?uin=uk.bl.ethos.259064}

\bibitem{Jung2018}
Jung, R., Krebbers, R., Jourdan, J.H., Bizjak, A., Birkedal, L., Dreyer, D.:
  {Iris from the ground up: A modular foundation for higher-order concurrent
  separation logic}. Journal of Functional Programming  (2018)

\bibitem{KleinEHACDEEKNSTW09}
Klein, G., Elphinstone, K., Heiser, G., Andronick, J., Cock, D., Derrin, P.,
  Elkaduwe, D., Engelhardt, K., Kolanski, R., Norrish, M., Sewell, T., Tuch,
  H., Winwood, S.: sel4: formal verification of an {OS} kernel. In: Matthews,
  J.N., Anderson, T.E. (eds.) Symposium on Operating Systems Principles (SOSP).
  pp. 207--220. {ACM} (2009)

\bibitem{Koh0LXBHMPZ19}
Koh, N., Li, Y., Li, Y., Xia, L., Beringer, L., Honor{\'{e}}, W., Mansky, W.,
  Pierce, B.C., Zdancewic, S.: From {C} to interaction trees: specifying,
  verifying, and testing a networked server. In: Mahboubi, A., Myreen, M.O.
  (eds.) Certified Programs and Proofs (CPP). pp. 234--248. {ACM} (2019)

\bibitem{KraglQH20}
Kragl, B., Qadeer, S., Henzinger, T.A.: Refinement for structured concurrent
  programs. In: Lahiri, S.K., Wang, C. (eds.) Computer Aided Verification
  (CAV). Lecture Notes in Computer Science, vol. 12224, pp. 275--298. Springer
  (2020)

\bibitem{Lamport2002}
Lamport, L.: Specifying Systems, The {TLA+} Language and Tools for Hardware and
  Software Engineers. Addison-Wesley (2002)

\bibitem{Leino10}
Leino, K.R.M.: Dafny: An automatic program verifier for functional correctness.
  In: Clarke, E.M., Voronkov, A. (eds.) Logic for Programming, Artificial
  Intelligence, and Reasoning (LPAR). Lecture Notes in Computer Science,
  vol.~6355, pp. 348--370. Springer (2010)

\bibitem{Leroy06}
Leroy, X.: Formal certification of a compiler back-end or: programming a
  compiler with a proof assistant. In: Morrisett, J.G., Jones, S.L.P. (eds.)
  Principles of Programming Languages (POPL). pp. 42--54. {ACM} (2006)

\bibitem{LesaniBC16}
Lesani, M., Bell, C.J., Chlipala, A.: Chapar: certified causally consistent
  distributed key-value stores. In: Bod{\'{\i}}k, R., Majumdar, R. (eds.)
  Principles of Programming Languages (POPL). pp. 357--370. {ACM} (2016)

\bibitem{nfm20-liu}
Liu, S., Sandur, A., Meseguer, J., {\"{O}}lveczky, P.C., Wang, Q.: Generating
  correct-by-construction distributed implementations from formal maude
  designs. In: Lee, R., Jha, S., Mavridou, A. (eds.) {NASA} Formal Methods.
  Lecture Notes in Computer Science, vol. 12229, pp. 22--40. Springer (2020)

\bibitem{LorchCKPQSWZ20}
Lorch, J.R., Chen, Y., Kapritsos, M., Parno, B., Qadeer, S., Sharma, U.,
  Wilcox, J.R., Zhao, X.: Armada: low-effort verification of high-performance
  concurrent programs. In: Donaldson, A.F., Torlak, E. (eds.) Programming
  Language Design and Implementation (PLDI). pp. 197--210. {ACM} (2020)

\bibitem{MuellerSchwerhoffSummers16}
M{\"u}ller, P., Schwerhoff, M., Summers, A.J.: Viper: A verification
  infrastructure for permission-based reasoning. In: Jobstmann, B., Leino,
  K.R.M. (eds.) Verification, Model Checking, and Abstract Interpretation
  (VMCAI). LNCS, vol.~9583, pp. 41--62. Springer (2016)

\bibitem{OHearn04}
O'Hearn, P.W.: Resources, concurrency and local reasoning. In: Gardner, P.,
  Yoshida, N. (eds.) Concurrency Theory (CONCUR). Lecture Notes in Computer
  Science, vol.~3170, pp. 49--67. Springer (2004)

\bibitem{OortwijnH19}
Oortwijn, W., Huisman, M.: Practical abstractions for automated verification of
  message passing concurrency. In: Ahrendt, W., Tarifa, S.L.T. (eds.)
  Integrated Formal Methods (iFM). Lecture Notes in Computer Science, vol.
  11918, pp. 399--417. Springer (2019)

\bibitem{OwickiG76}
Owicki, S.S., Gries, D.: Verifying properties of parallel programs: An
  axiomatic approach. Commun. {ACM}  \textbf{19}(5),  279--285 (1976).
  \doi{10.1145/360051.360224}, \url{https://doi.org/10.1145/360051.360224}

\bibitem{Penninckx0P15}
Penninckx, W., Jacobs, B., Piessens, F.: Sound, modular and compositional
  verification of the input/output behavior of programs. In: Vitek, J. (ed.)
  European Symposium on Programming (ESOP). Lecture Notes in Computer Science,
  vol.~9032, pp. 158--182. Springer (2015)

\bibitem{RahliVVV18}
Rahli, V., Vukotic, I., V{\"{o}}lp, M., Ver{\'{\i}}ssimo, P.J.E.: Velisarios:
  Byzantine fault-tolerant protocols powered by coq. In: Ahmed, A. (ed.)
  European Symposium on Programming (ESOP). Lecture Notes in Computer Science,
  vol. 10801, pp. 619--650. Springer (2018)

\bibitem{Reynolds2002}
Reynolds, J.C.: {Separation logic: A logic for shared mutable data structures}.
  pp. 55--74 (2002)

\bibitem{PintoDG14}
da~Rocha~Pinto, P., Dinsdale{-}Young, T., Gardner, P.: {TaDA}: {A} logic for
  time and data abstraction. In: European Conference on Object-Oriented
  Programming (ECOOP). Lecture Notes in Computer Science, vol.~8586, pp.
  207--231. Springer (2014)

\bibitem{SergeyNB15}
Sergey, I., Nanevski, A., Banerjee, A.: Mechanized verification of fine-grained
  concurrent programs. In: Grove, D., Blackburn, S.M. (eds.) Programming
  Language Design and Implementation (PLDI). pp. 77--87. {ACM} (2015)

\bibitem{SergeyWT18}
Sergey, I., Wilcox, J.R., Tatlock, Z.: Programming and proving with distributed
  protocols. {PACMPL}  \textbf{2}({POPL}),  28:1--28:30 (2018)

\bibitem{Igloo}
Sprenger, C., Klenze, T., Eilers, M., Wolf, F.A., M\"uller, P., Clochard, M.,
  Basin, D.: Igloo: Soundly linking compositional refinement and separation
  logic for distributed system verification. In: Object-Oriented Programming
  Systems, Languages, and Applications (OOPSLA). vol.~4. ACM (2020)

\bibitem{Trillium}
Timany, A., Gregersen, S.O., Stefanesco, L., Gondelman, L., Nieto, A.,
  Birkedal, L.: Trillium: Unifying refinement and higher-order distributed
  separation logic. CoRR  \textbf{abs/2109.07863} (2021)

\bibitem{Vafeiadis11}
Vafeiadis, V.: Concurrent separation logic and operational semantics. In:
  {MFPS}. Electronic Notes in Theoretical Computer Science, vol.~276, pp.
  335--351. Elsevier (2011)

\bibitem{WoosWATEA16}
Woos, D., Wilcox, J.R., Anton, S., Tatlock, Z., Ernst, M.D., Anderson, T.E.:
  Planning for change in a formal verification of the {Raft} consensus
  protocol. In: Avigad, J., Chlipala, A. (eds.) Certified Programs and Proofs
  (CPP). pp. 154--165 (2016)

\bibitem{XiaZHHMPZ20}
Xia, L., Zakowski, Y., He, P., Hur, C., Malecha, G., Pierce, B.C., Zdancewic,
  S.: Interaction trees: representing recursive and impure programs in coq.
  Proc. {ACM} Program. Lang.  \textbf{4}({POPL}),  51:1--51:32 (2020)

\end{thebibliography}
